\documentclass[preprint,10pt]{elsarticle_mod}
\usepackage{url}
\usepackage{amsmath}
\usepackage{amssymb}
\usepackage{amsthm}
\usepackage{graphics}
\usepackage{graphicx}
\usepackage{multicol}
\usepackage{times}
\usepackage{algorithm,algorithmic}

\newtheorem{theorem}{Theorem}

\newtheorem{Definition}[theorem]{Definition}
\newtheorem{Lemma}[theorem]{Lemma}

\newtheorem{Claim}[theorem]{Claim}

\newcommand{\SWITCH}[1]{\STATE \textbf{switch} (#1)}
\newcommand{\ENDSWITCH}{\STATE \textbf{end switch}}

\newcommand{\INITIALIZATION}{\hskip -0.55cm \textbf{Initialization:\,\,\,}}
\newcommand{\CASE}[1]{\STATE \textbf{case} #1\textbf{} \begin{ALC@g}}
\newcommand{\ENDCASE}{\end{ALC@g}}

\urldef{\mailsa}\path|rafal.kapelko@pwr.edu.pl|

\begin{document}
\begin{frontmatter}

\title{Analysis of the Threshold for Energy Consumption in Displacement of Random Sensors}
\author[pwr]{Rafa\l{} Kapelko\fnref{pwrfootnote}\corref{cor1}}
\ead{rafal.kapelko@pwr.edu.pl}
\cortext[cor1]{Corresponding author at: Department of Fundamentals of Computer Science,
Wroc{\l}aw University of Science and Technology, 
 Wybrze\.{z}e Wyspia\'{n}skiego 27, 50-370 Wroc\l{}aw, Poland. Tel.: +48 71 320 30 48; fax: +48 71 320 07 51.}
\address[pwr]{ Department of Fundamentals of Computer Science, 
Wroc{\l}aw University of Science and Technology, Poland}
\fntext[pwrfootnote]{Supported by Polish National Science Center (NCN) grant 2019/33/B/ST6/02988}
\begin{abstract}
The fundamental problem of energy-efficient reallocation of mobile random sensors to provide full coverage without interference is addressed in this paper.
We  consider $n$ mobile sensors with identical sensing range placed randomly  on the unit interval and on the unit square.
The main contribution is summarized as follows:
\begin{itemize}
\item If the sensors are placed on the unit interval we explain \textbf{a increase}
around the sensing radius equal to $\frac{1}{2n}$  and the interference distance equal to $\frac{1}{n}$
for the expected minimal $a$-total displacement, 
\item If the sensors are placed on the unit square we explain \textbf{a increase}
around the square sensing radius equal to $\frac{1}{2 \sqrt{n}}$ and the interference distance equal to $\frac{1}{\sqrt{n}}$ 
for the expected minimal $a$-total displacement. 
\end{itemize}
\end{abstract}

\end{frontmatter}

\section{Introduction}
Mobile sensors are being deployed in many application areas to enable easier information retrieval in the communication environments,
from sensing and diagnostics to critical infrastructure monitoring (e.g. see  \cite{siamcontrol2015, sajal2008, mohamed2017}). 
Current reduction in manufacturing costs makes random deployment
of the sensors more attractive. Even existing sensor placement schemes cannot guarantee precise placement of sensors, 
so their initial deployment may be somewhat random.

A typical sensor is able to sense and thus cover a bounded region specified by its sensing radius \cite{kumar2005}. To monitor and protect 
a larger region against intruders every point of the region has to  be within 
the sensing range of a sensor. It is also known that proximity between sensors affects the transmission and reception of signals and causes
the degradation of performance \cite{gupta}. Therefore in order  to avoid interference a critical value, say $s$ is established.
It is assumed that for a given parameter $s$ two sensors interfere with each other during communication if their distance is
less than $s$ (see \cite{pervasiveKAPELKO, kranakis_shaikhet}). However, random deployment of the sensors might leave some gaps in the coverage of the area
and the sensors may be too close to each other. Therefore, to attain coverage of the area and to avoid interference
the reallocation of sensors may be the only option. Moreover, the ability to move the mobile sensors to the final destinations is not unrealistic.
Clearly, the displacement of a team of sensors should be performed in the most efficient way.

The \textit{energy consumption} for the displacement of a set of $n$ sensors is measured by the sum of the respective displacements to the power of the individual sensors.
We define below the concept of \textit{$a$-total displacement}. 
\begin{Definition}[$a$-total displacement] 
\label{def:atotalasam}
Let $a> 0$ be a constant.
Suppose the displacement of the $i$-th sensor is a distance $d_i$. The $a$-total 
displacement is defined as the sum $\sum_{i=1}^n d_i^a$. 
\end{Definition}
Motivation for this cost metric arises from the fact that the parameter $a$ in the exponents represents
various conditions  on the region lubrication and friction which  affect the sensor movement.


We  consider $n$ mobile sensors which are  placed independently and uniformly at random 
on the unit interval and on the unit square. 

For the case of unit interval $[0,1]$ each sensor is equipped with an omnidirectional antenna of identical \textit{sensing radius} $r_1>0.$ Thus, a sensor placed at location $x$ on the unit interval 
can cover any point at distance at most $r_1$ either to the left or right of $x.$ (See Figure \ref{fig.1}(a)).

For the case of unit square $[0,1]^2$ each sensor has identical \textit{square sensing radius} $r_2>0.$
\begin{Definition}[cf. \cite{adhocnow2015_KK} Square Sensing Radius]
\label{def:square}
We assume that a sensor located in position $(x_1,x_2)$ where
$0 \le x_1,x_2\le 1$ 
can cover any point in the area delimited by the square with corner points
$(x_1\pm r_2,x_2\pm r_2)$ and call $r_2$ the square sensing radius of the sensor.
\end{Definition}
Figure \ref{fig.1}(b) illustrates the \textit{square sensing radius}.
\begin{figure*}[ht]
\label{fig.1}
  \begin{center}
    \includegraphics[width=0.9\textwidth]{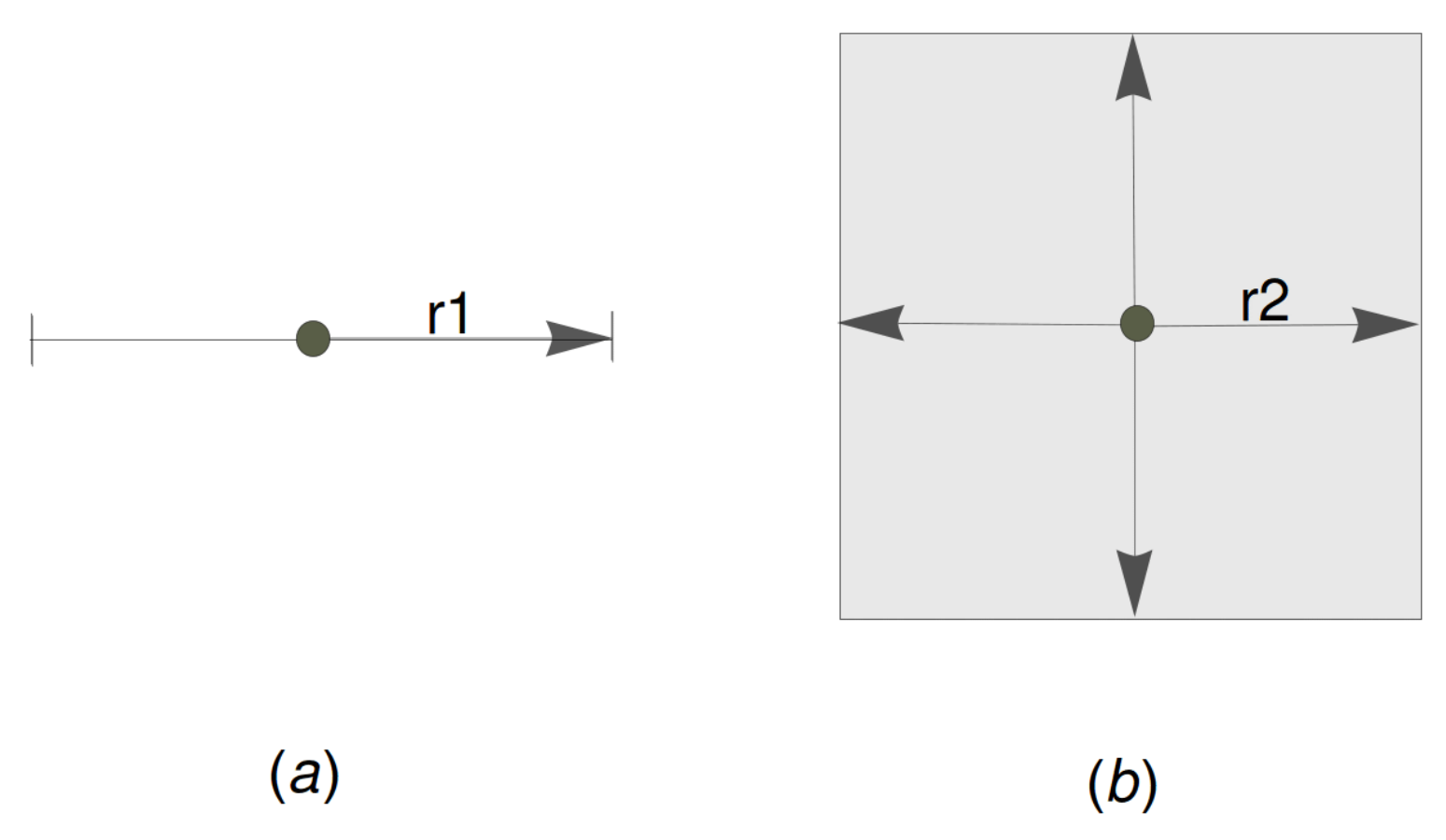}
  \end{center}
  \caption{(a) sensing radius $r_1$ on a line. (b) square sensing radius $r_2.$}
    \label{fig:anna1}
\end{figure*}
However, in most cases the sensing area of a sensor is a circular disk of radius $r_c$ but 
our upper bound result proved in the sequel for square sensing radius $r_2$ are obviously valid
for circular disk of radius $r_c$ equal to $\sqrt{2}r_2$ circumscribing the square.

The sensors are required to move from their current random locations to new positions so as to satisfy
the following requirement.
\begin{Definition}[$(r_m,s)$-C\&I requirement] Fix $m\in\{1,2\}.$ A set of sensors placed on the $m$-dimensional unit cube satisfy 
the $(r_m,s)$ coverage \& interference requirement if: 
\label{def:anna}
requires:
\begin{enumerate}
\item[(a)] Every point on the  $m$-dimensional unit cube $[0,1]^m$ is within the range $r_m$ of a sensor, i.e. the $m$-dimensional unit cube is completely covered.
\item[(b)] Each pair of sensors is placed  at Euclidean distance greater or equal to $s.$ 
\end{enumerate}
\end{Definition}
In this paper we investigate the problem of energy efficient displacement of the random mobile sensors. 
\begin{Definition}[energy efficient displacement]
Assume that $n$ mobile sensors are  placed independently and uniformly at random 
on the unit interval or on the unit square. The sensors move from their intial current location to the final
destination so that in their
final placement the sensor system satisfy the $(r_m,s)$-\textit{coverage \& interference} requirement
and the $a$-total displacement is minimized in expectation.
\end{Definition}
Throughout the paper, we will use
the Landau asymptotic notations: 
\begin{itemize}
\item[(i)] $f(n)=O(g(n))$ if there exists a constant $C_1>0$ and integer $N$ such that $|f(n)|\le C_1|g(n)|$ for all $n>N,$
\item[(ii)] $f(n)=\Omega(g(n))$ if there exists a constant $C_2>0$ and integer $N$ such that $|f(n)|\ge C_2|g(n)|$ for all $n>N,$
\item[(iii)] $f(n)=\Theta(g(n))$ if and only if $f(n)=O(g(n))$ and $f(n)=\Omega(g(n)),$
\end{itemize}

\subsection{Contribution and Outline of the Paper}
Let $a>0$ be a constant. Assume that $n$ mobile sensors with identical sensing radius $r_1$ and square sensing radius $r_2$ are placed independently at random 
with the uniform distribution on the unit interval and on the unit square. 

In this paper we give \textit{the picture} of \textbf{the threshold phenomena}
for the \textit{coverage} \& \textit{interference} requirement in one dimension, as well as in two dimension
(see Definition \ref{def:anna}). The $a$-total displacement (the energy consumption)
is used to measure the movement cost (see Definition \ref{def:atotalasam}) while Euclidean distance is used for
the interference distance and the sensing area of a sensor in two dimension is a square
(see Definition \ref{def:square}). Let us also recall that in two dimension 
the sensors can move \textit{directly to the final locations via the shortest route} not only in vertical and horizontal fashion.

 Let $\epsilon>0$, $1>\delta>0$ be \textit{arbitrary small constants} independent 
on the number of sensors $n.$ 

Table $1$ summarizes our main contribution in one dimension. 
\begin{table*}[ht]
\caption {The expected minimal $a$-total displacement of $n$ random sensors on the unit interval $[0,1]$ as a function of the sensing radius $r_1$ and the interference value $s,$ where $\epsilon>0$, $1>\delta >0.$ } \label{tab:line} 
\begin{center}
 \begin{tabular}{|*{4}{c|}} 
 \hline
  \begin{tabular}{c}Sensing\\ radius $r_1$\end{tabular} &\begin{tabular}{c}Interference\\ distance $s$ \end{tabular}  &  \begin{tabular}{c}Expected minimal $a$-total\\ displacement for\\ $(r_1,s)-C\&I$ \textit{requirement}\end{tabular} & Theorem \\ [0.5ex]
   \hline 
  $r_1=\frac{1}{2n}$ & $s=\frac{1}{n}$  
  &\begin{tabular}{c}$\frac{\Gamma(\frac{a}{2}+1)}{2^{\frac{a}{2}}(1+a)}n^{1-\frac{a}{2}}+O\left(n^{-\frac{a}{2}}\right),$\\
  $a> 0$ \end{tabular} &  
  \ref{thm:exact_one}[cf. \cite{fuchs2020}] 
  \\ [0.5ex]
 \hline 
\begin{tabular}{c}$r_1=\frac{1+\epsilon}{2n},$\\$\epsilon >0$\end{tabular} &  \begin{tabular}{c}$s=\frac{1-\delta}{n},$\\$1>\delta >0$\end{tabular} & 
 $O\left(n^{1-a}\right),\,\,$ $a>0 $ &  \ref{thm:const} \\  [0.5ex]
 \hline
 \end{tabular} 
\end{center}
\end{table*} 

As the sensing radius $r_1$ increases from $\frac{1}{2n}$ to $\frac{1+\epsilon}{2n}$ and  the interference distance $s$ decreases from $\frac{1}{n}$ to $\frac{1-\delta}{n}$
there is a decline 
from $\Theta\left({n^{\frac{a}{2}}}n^{1-a}\right)$
to $O\left(n^{1-a}\right)$
in the expected minimal $a$-total displacement for all powers $a>0.$

Table $2$ summarizes our main contribution in two dimensions. 
\begin {table*}[ht]
\caption {The expected minimal $a$-total displacement of $n$ random sensors on the unit square $[0,1]^2$ as a function of the square sensing radius $r_2$ and the interference value $s,$ where
$\epsilon>0$, $1>\delta>0.$} \label{tab:square} 
\begin{center}
 \begin{tabular}{|*{4}{c|}} 
 \hline
  \begin{tabular}{c}Square\\ sensing\\ radius $r_2$\end{tabular} & \begin{tabular}{c}Interference\\ distance $s$\end{tabular}   &  \begin{tabular}{c}Expected minimal $a$-total\\ displacement for\\ $(r_2,s)-C\&I$ 
 \textit{requirement}\end{tabular} & Theorem \\ [0.5ex]
   \hline 
  \begin{tabular}{c}$r_2=\frac{1}{2\sqrt{n}}$\end{tabular}  & \begin{tabular}{c}$s=\frac{1}{\sqrt{n}}$\end{tabular} 
  & \begin{tabular}{c}$\Theta\left(\sqrt{\ln(n)n}\right)$ \,\,if\,\, $a=1$\\ $\Omega\left((\ln(n))^{\frac{a}{2}}n^{1-\frac{a}{2}}\right)$ \,\,if\,\, $a> 1$\end{tabular}&  \begin{tabular}{c} \ref{thm:tal}[cf.\cite{talagrand_2014}]  \\\ref{thm:tal_a} \end{tabular}  \\ [0.5ex]
 \hline 
 \begin{tabular}{c}$r_2=\frac{1+\epsilon}{2\lfloor\sqrt{n}\rfloor},$\\$\epsilon >0$\end{tabular} & \begin{tabular}{c}$s=\frac{1-\delta}{\lfloor\sqrt{n}\rfloor},$\\$1>\delta >0$\end{tabular} & $O\left(n^{1-\frac{a}{2}}\right)$\,\,if\,\, $a> 0$ &  \ref{thm:const2}  
 \\  [0.5ex]
 \hline
 \end{tabular}
\end{center}
\end{table*}

As the square sensing radius $r_2$ increases from $\frac{1}{2 \sqrt{n}}$ to $\frac{1+\epsilon}{2\lfloor\sqrt{n}\rfloor}$
and  the interference distance $s$ decreases from $\frac{1}{\sqrt{n}}$ to $\frac{1-\delta}{\lfloor\sqrt{n}\rfloor}$
There is a decline 
from $\Omega\left({{({\ln(n)})^{{\frac{a}{2}}}}}n^{1-\frac{a}{2}}\right)$ to $O\left(n^{1-\frac{a}{2}}\right)$
in the expected minimal $a$-total displacement for all powers $a\ge 1.$

Notice that $n$ sensors on the unit interval $[0,1]$ with sensing radius $r_1=\frac{1}{2n}$ and the interference distance  have to move to 
the anchor positions to satisfy $\left(r_1,s\right)$-\textit{coverage \& interference}. When $r_1>\frac{1}{2n}$ and $s<\frac{1}{n}$
there are no anchors positions predetermined in advance.  The similar remark holds for the sensors on the unit square $[0,1]^2.$

Our theoretical results imply that the expected $a$-total displacement  is constant and independent on number of sensors  for some parameters $a.$ 
Namely, we have the following upper bounds:
\begin{itemize}
\item[(i)] For the random sensors on the unit interval, when $$n(2r_1)=1+\epsilon,$$ i.e. the sum of sensing area of $n$ sensors is a \textit{little bigger} then the \textit{length of the unit interval},
it is possible to provide the full area coverage in $O(1)$ expected $a$-total displacement with $a\ge 1.$
\item[(ii)]
For the random sensors on the unit square, when $$n(2r_2)^2\sim (1+\epsilon)^2\,\,\,\,\,\text{as}\,\,\,\,\,n\rightarrow\infty,$$ 
i.e. the sum of sensing area of $n$ sensors is asymptotically a \textit{little bigger} then the \textit{area of unit square},
the expected $a$-total displacement with $a\ge 2$ to provide full area coverage is $O(1).$ 
Obviously, this result is easily applicable to the model when the sensing area of a sensor is a circular disk of radius $r_c$ by taking circle circumscribing the square. 
Namely, when 
$$n\pi (r_c)^2\sim \frac{\pi}{2}(1+\epsilon)^2\,\,\,\,\,\text{as}\,\,\,\,\,n\rightarrow\infty$$ 
then the expected $2$-total displacement to provide full area coverage is constant. 
\end{itemize}
This \textbf{constant cost} seems to be of \textbf{practical importance} due to efficient monitoring against illegal trespassers.
It is well known that intrusion detection is an important application of wireless sensor networks. In this case it is necessary to ensure coverage with good communication.

Notice that constant expected cost in (i) and (ii) are valid for $n$ random sensors with identical sensing radius $r_1=\frac{x(1+\epsilon)}{2n}$ on the interval of length $x$
and for $n$ random sensors with identical square sensing radius $r_2=\frac{x(1+\epsilon)}{2\lfloor\sqrt{n}\rfloor}$ on the square $[0,x]\times[0,x].$

We also present \textit{3 algorithms} (see Algorithms (\ref{alg_left}-\ref{alg_coverage2})). It is worthwhile to mention that, even though the algorithms are simple 
the analysis is challenging. Notice that Algorithms (\ref{alg_left}-\ref{alg_coverage2}) can be implemented by a centralized controller telling each sensor where
and when to move. In Section \ref{sec:preliminary} we prove some technical properties of Beta distribution with special positive integer parameters
needed in the current paper (see Lemma \ref{lemma_e} and Lemma \ref{lemma_f}).

The overall organization of the paper is as follows. Subsection \ref{sub:relate} briefly summarizes some related work. In Section \ref{sec:preliminary} we present some preliminary results that will be used in the sequel. 
Sections \ref{sec:one1} and \ref{sec:one} deals with sensors on the unit interval. In Sections \ref{sec:alpha} and \ref{sec:two} we investigate sensors on the unit square, 
while
further insights in the higher dimension are discussed in
Section \ref{sec:discussion}.
Section \ref{sec:simu} deals with experimental evaluation of Algorithm \ref{alg_left}. 
Section \ref{sec:con} contains conclusions and directions for future work.
Finally, for the sake of readability, certain technical proofs are defarred to the Appendices.
\subsection{Related Work}
\label{sub:relate}
There are extensive studies dealing with both coverage (e.g., see \cite{abbasi2009movement, salajkcover, tcs2009, saipulla2009}) and interference problems 
(e.g., see \cite{burkh02, morin2018, halld07, Kranakis2010}). Closely related to barrier and area coverage the matching problem is also of interest in the research community
(e.g., see \cite{ajtai_84, milos, kapelkogamma, talagrand_2014})

An important setting in considerations for coverage of a domain is when the sensors are initially placed at random with the uniform distribution. 
Some authors proposed using several rounds of random displacement to achieve complete coverage of a domain \cite{eftekhari13, yan}. 
Another approach is to have the sensors relocate from their initial position to a new position to achieve the desired coverage \cite{MinMax,eftekhari13d}. 

More importantly, our work is closely related to the papers \cite{KK_2016_cube, kapelkokranakisIPL} in respect to analysis of the expected $a$-total displacement 
for coverage problem where the sensors are randomly placed on the unit interval \cite{kapelkokranakisIPL} and in the higher dimension \cite{KK_2016_cube}. 
Both papers study performance bounds for some algorithms, using Chernoff's inequality.
The methods used in these papers have limitations - the most important and difficult cases when the sensing radius $r_1$ is close to $\frac{1}{2n}$ 
and
the square sensing radius $r_2$ is close to $\frac{1}{2\sqrt{n}}$ were not included in \cite{KK_2016_cube, kapelkokranakisIPL}. Moreover, in the paper 
\cite{KK_2016_cube} the sensors can move only in parallel to the axes. 
Hence, the analysis of coverage problem in \cite{KK_2016_cube} is incomplete. Moreover, it is natural to investigate the general case when the sensor can move
directly to the final locations via the shortest route not only in vertical and horizontal fashion. 
\textbf{The novelty of work} in the current paper lies in studying the cases for the threshold phenomena, when the sensing radius $r_1$ is close to $\frac{1}{2n},$ i.e. 
$r_1=\frac{1+\epsilon}{2n}$ and
the square sensing radius $r_2$ is close to $\frac{1}{2\sqrt{n}},$ i.e. $r_2=\frac{1+\epsilon}{2\lfloor\sqrt{n}\rfloor}$
for coverage \& interference, provided that $\epsilon$ is an  arbitrary small constant independent on the number of sensors $n.$ 
Compared to the coverage problem,  $(r_m,s)-C\&I$ \textbf{requirement} not only ensures coverage, but also  avoids interference
and is more reasonable in order to provide reliable communication within the network. 

Finally, it is worth mentioning that, our work is related  to the series of papers \cite{ICDCNkapelko,pervasiveKAPELKO,ICDCN2020kapelko,transdaskapelko2021}.
In \cite{ICDCNkapelko,pervasiveKAPELKO} the author 
investigated the maximum of the expected sensor's displacement (the time required) for coverage \& interference. In \cite{ICDCNkapelko,pervasiveKAPELKO} it is assumed that
the $n$ sensors are initially deployed on the $[0,\infty)$ according to the arrival times of the Poisson process with arrival rate $\lambda>0$
and coverage (connectivity) is in the sense that there are no uncovered points from the origin to the last rightmost sensor. 
The work by \cite{ICDCN2020kapelko} investigates the expected minimal $a$-total discplacement for interference-connectivity requirement
when the $n$ sensors are initially placed on the $[0,\infty)^d$ according to $d$ identical and independent Poisson processes 
each with arrival rate $\lambda>0.$
It is worth pointing out that the $d$-dimensional model in \cite{ICDCN2020kapelko} is only the
direct extension of the interference-connectivity requirement from
one dimension to the $d$-dimensional space and the sensors move 
only in parallel to the axes.
\section{Preliminaries}
\label{sec:preliminary}
In this section we introduce some basic concepts and notations that will be used in the sequel.
We also present three lemmas which will be helpful in proving our main results. 
In this paper, in the one dimensional scenario, the $n$ mobile sensors are thrown independently at random
following the uniform distribution in the unit interval $[0, 1].$
Let $X_{(\ell)}$ be the position of the $\ell$-th sensor  after sorting the initial random locations of $n$ sensors with respect to the origin of the interval $[0,1],$
i.e. the $\ell$-th order statistics of the uniform distribution in the unit interval. 
It is known that the random variable $X_{(\ell)}$ obeys the Beta distribution with parameters $\ell,n+1-\ell$ (see \cite[page13]{arnold2008}). 

Assume that $c,d$ are positive integers. The Beta distribution $\mathrm{Beta}(c,d)$ (see \cite{NIST}) with parameters $c,d$ is the continuous distribution on $[0,1]$ with 
probability density function $f_{c,d}(t)$ given by
\begin{equation}
\label{eq:prosto}
f_{c,d}(t)=c\binom{c+d-1}{c}t^{c-1}(1-t)^{d-1},\,\,\, \text{when}\,\,\,0\le t\le 1.
\end{equation}

The cumulative distribution function of the Beta distribution with parameters $c,d$ is
given by the incomplete Beta function
\begin{equation}
\label{incomplete:first}
I_z(c,d)=c\binom{c+d-1}{c}\int_{0}^{z}t^{c-1}(1-t)^{d-1}dt\,\,\, \text{for}\,\,\,0\le z\le 1.
\end{equation}

Moreover, the incomplete Beta function is related to the binomial distribution by
\begin{equation}
\label{probal_eq}
1-I_z(c,d)=\sum_{j=0}^{c-1}\binom{c+d-1}{j}z^j(1-z)^{c+d-1-j}
\end{equation}
(see \cite[Identity 8.17.5]{NIST} for $c:=m,$ $d:=n-m+1$ and $x:=z$) and the binomial identity 
\begin{equation}
\label{eq:binole}
\sum_{j=0}^{c+d-1}\binom{c+d-1}{j}z^j(1-z)^{c+d-1-j}=1.
\end{equation}
The following inequality which relates binomial and Poisson distribution was discovered by Yu. V. Prohorov (see  \cite[Theorem 2]{LeCam1965}, \cite{prohorov}).
\begin{equation}
\label{eq:lecam}
\binom{n}{j}x^j(1-x)^{n-j}\le \left(\frac{n}{m_1}\right)^{\frac{1}{2}} e^{-nx}\frac{(nx)^j}{j!},
\end{equation}
where $m_1$ is some integer which satisfies $n(1-x)-1<m_1\le n(1-x).$

We will also use the classical Stirling's approximation for factorial (see \cite[page 54]{feller1968})
\begin{equation}
\label{eq:stirlingform}
\sqrt{2\pi}N^{N+\frac{1}{2}}e^{-N+\frac{1}{12N+1}}< N!< \sqrt{2\pi}N^{N+\frac{1}{2}}e^{-N+\frac{1}{12N}}.
\end{equation}
We use the following notation
$|x|^{+}=\max\{x,0\}$ 
for the positive parts of $x\in\mathbb{R}.$

We are now ready to give some useful properties of Beta distribution in the following sequences of lemmas.
\begin{Lemma}
 \label{lemma_first}
Let $a>0.$ Assume that $n$ is positive integer.
Then
$$
\Pr\left[\mathrm{Beta}(n,1) <1-\frac{1}{n^{\frac{a}{1+a}}}\right]<\frac{1}{e^{n^{\frac{1}{1+a}}}}.
$$
\end{Lemma}
\begin{proof}  
 First of all observe that (see (\ref{eq:prosto}) for $c:=n$ and $d:=1.$)
\begin{align}
\label{eq:dus_b}
\nonumber\Pr\left[\mathrm{Beta}(n,1)<1-\frac{1}{n^{\frac{a}{1+a}}}\right]&=\int_{0}^{1-\frac{1}{n^{\frac{a}{1+a}}}}f_{n,1}(t)dt=\left(1-\frac{1}{n^{\frac{a}{1+a}}}\right)^n\\
&=\left(\left(1-\frac{1}{n^{\frac{a}{1+a}}}\right)^{n^{\frac{a}{a+1}}}\right)^{n^{\frac{1}{a+1}}}.
\end{align}
Using (\ref{eq:dus_b}) and the basic inequality $(1-x)^{1/x}<e^{-1}$ when $x>0$ we have
$$\Pr\left[\mathrm{Beta}(n,1)<1-\frac{1}{n^{\frac{a}{1+a}}}\right]<\frac{1}{e^{n^{\frac{1}{1+a}}}}$$
which completes the proof. 
\end{proof} 
\begin{Lemma}
 \label{lemma_e}
Let $a>0$ be a constant.
Fix $\gamma>0$ independent on $n.$ Let $\rho=\frac{1+\gamma}{n}.$
Assume that $\ell, n$ are positive integers and $\ell\le n.$
Then
\begin{equation}
\label{eq:estaax}
\mathbb{E}\left[\left(|\mathrm{Beta}(\ell,n-\ell+1)-\rho \ell|^{+}\right)^a\right]=O\left(\frac{1}{n^a}\right),
\,\,\,\text{uniformly in}\,\,\,\ell\in\{1,2,\dots,n\},
\end{equation}
\begin{equation}
\label{eq:estaay}
\sum_{\ell=1}^{n}\frac{n}{\ell} \mathbb{E}\left[\left(|\mathrm{Beta}(\ell,n-\ell+1)-\rho \ell|^{+}\right)^a\right]=O\left(n^{1-a}\right).
\end{equation}
\end{Lemma}  
\begin{proof}
The proof is given in Appendix A. 
\end{proof}
\begin{Lemma}
 \label{lemma_f}
Let $a>0$ be a constant.
Fix $1>\delta>0$ independent on $n.$ Let $s=\frac{1-\delta}{n}.$
Assume that $\ell, n$ are positive integers and $\ell\le n.$
Then
\begin{equation}
\label{eq:estaaybn}
\sum_{\ell=1}^{n}\frac{n}{\ell} \mathbb{E}\left[\left(|s\ell-\mathrm{Beta}(\ell,n-\ell+1)|^{+}\right)^a\right]=O\left(n^{1-a}\right).
\end{equation}
\end{Lemma}
\begin{proof}
The proof is given in Appendix B. 
\end{proof}
The following lemma will simplify the upper bound estimations in Section \ref{sec:one} and Section \ref{sec:two}.
\begin{Lemma}
\label{lem:cruciana}
Fix $a> 0.$ Assume that the sensor movement $M$ is the finite sum of movements $M_i$ for $i=1,2,\dots, l,$ i.e.
$M=\sum_{i=1}^{\ell}M_i.$
Then 
$$\mathbb{E}[M^a]\le C_{a,\ell}\sum_{i=1}^{\ell}\mathbb{E}[M_i^a],$$
where $C_{a,\ell}$ is some constant which depend only on fixed $a$ and $\ell.$
\end{Lemma}
\begin{proof} 
Firstly we recall two elementary inequalities.

Fix $a\ge 1.$ Let $x,y \ge 0.$ Then
\begin{equation}
\label{eq:duszal}
(x+y)^a\le 2^{a-1}(x^a+y^a).
\end{equation}
Notice that Inequality (\ref{eq:duszal}) is the consequence of the fact that $f(x)=x^a$ is convex over $\mathbf{R_+}$ for $a\ge 1.$

Fix $a\in(0,1).$ Let $x,y \ge 0.$ Then
\begin{equation}
\label{eq:duszalb}
(x+y)^a\le x^a+y^a.
\end{equation}
Combining together Inequality (\ref{eq:duszal}) and Inequality (\ref{eq:duszalb}) for the sum $\sum_{i=1}^{\ell}M_i$ and passing to the expectations we derive 
$$\mathbb{E}[M^a]\le C_{a,\ell}\sum_{i=1}^{\ell}\mathbb{E}[M_i^a].$$
This proves Lemma \ref{lem:cruciana}.
\end{proof}
\section{Coverage \& interference requirement when the sensing radius $r_1=\frac{1}{2n}$ and the interference distance $s=\frac{1}{n}$}
\label{sec:one1}
In this section, we recall \textbf{the known results} about the expected $a$-total displacement to fulfill the
$(r_1,s)-C\&I$ \textit{requirement} when $n$ mobile sensors with identical sensing radius $r_1=\frac{1}{2n}$
are distributed uniformly at random and independently
on the unit interval $[0,1].$ That is, the sum of sensing area of $n$ sensors is \textbf{equal} to the \textbf{length of the unit interval.} Observe that in the case when
the sensing radius $r_1=\frac{1}{2n}$ and the interference distance $s=\frac{1}{n}$ 
the only way to achieve  $\left(r_1,s\right)$-\textit{coverage \& interference} requirement on the unit interval $[0,1]$ is for the sensors  to occupy the equidistant anchor positions
$\frac{i}{n} - \frac{1}{2n}$, for $i = 1, 2, \ldots , n.$
The following exact asymptotic result was proved in \cite{kapelkokranakisIPL}. 
\begin{theorem}[\cite{kapelkokranakisIPL}]
\label{thm:spaa}
Let $a$ be an even positive natural number.
Assume that, $n$ mobile sensors are thrown uniformly and independently at random on the unit interval $[0,1].$ The expected $a$-total displacement
of all $n$ sensors, when the $i$-th sensor sorted in increasing order moves 
from its current random location to the equidistant anchor location $\frac{i}{n} - \frac{1}{2n}$, for $i = 1, 2, \ldots , n$,  respectively, 
is $\frac{\left(\frac{a}{2}\right)!}{2^{\frac{a}{2}}(1+a)}n^{1-\frac{a}{2}}+O\left(n^{-\frac{a}{2}}\right).$
\end{theorem}
In \cite{fuchs2020}, Theorem \ref{thm:spaa} was extended to all real valued exponents $a>0.$
\begin{theorem}[\cite{fuchs2020}]
\label{thm:exact_one}
Fix $a>0.$ 
Assume that, $n$ mobile sensors are thrown uniformly and independently at random on the unit interval $[0,1].$ The expected $a$-total displacement
of all $n$ sensors, when the $i$-th sensor sorted in increasing order moves 
from its current random location to the equidistant anchor location $\frac{i}{n} - \frac{1}{2n}$, for $i = 1, 2, \ldots , n$,  respectively, 
is 
\begin{equation}
\label{eq:firsteq}
\frac{\Gamma\left(\frac{a}{2}+1\right)}{2^{\frac{a}{2}}(1+a)}n^{1-\frac{a}{2}}+O\left(n^{-\frac{a}{2}}\right).
\end{equation}
\end{theorem}
The gamma function $\Gamma(a)$ is defined to be an extension of the factorial to real number arguments. It is related to the factorial 
by $\Gamma\left(\frac{a}{2}+1\right)=\left(\frac{a}{2}\right)!$ provided that $\frac{a}{2}\in{\mathbb N}.$
It is also worthwhile to mention that, the extension of direct combinatorial method from \cite{kapelkokranakisIPL}
leads to exact asymptotic result in Theorem \ref{thm:exact_one} only when $a$ is an odd natural number (see \cite[Theorem 2]{kapelko_dmaa}).

\section{Coverage \& interference requirement when the square sensing radius $r_2=\frac{1}{2\sqrt{n}}$ and the interference distance $s=\frac{1}{\sqrt{n}}$}
\label{sec:alpha}
In this section, we analyze the expected $a$-total displacement to achieve
$(r_2,s)-C\&I$ \textit{requiremnt} when $n$ mobile sensors with identical square sensing radius $r_2=\frac{1}{2\sqrt{n}}$
are thrown uniformly at random and independently on the unit square $[0,1]^2,$ provided that $n$ is the square of a natural number.
That is, the sum of sensing area of $n$ sensors is \textbf{equal} to the \textbf{area of unit square.}

Observe that to fulfill $\left(\frac{1}{\sqrt{n}},\frac{1}{2\sqrt{n}}\right)$-\textit{coverage \& interference} requirement the sensors have to occupy the following anchor positions
$\left(\frac{k}{\sqrt{n}}-\frac{1}{2\sqrt{n}}, \frac{l}{\sqrt{n}}-\frac{1}{2\sqrt{n}}\right),$
where $1\le k,l\le \sqrt{n}$ and $n$ must be \textit{the square of a natural number.} 

It is known that expected $1$-total displacement in this case is $\Theta\left(\sqrt{\ln(n)n}\right).$
Namely, the following theorem about \textit{the optimal transportation cost for random matching} \textbf{was obtained} in \cite{talagrand_2014} a book related to these problems which develops modern methods to bound stochastic processes.
\begin{theorem}[\cite{talagrand_2014}, Chapter 4.3]  
\label{thm:tal}
Let $n=q^2$ for some $q\in\mathbb{N}.$ Assume that $n$  mobile sensors $X_1,X_2,\dots,X_n$ are thrown uniformly and independently at random on the unit square $[0,1]^2.$ 
Consider the non-random points $(Z_i)_{i\le n}$ evenly distributed as follows:
$Z_i=\left(\frac{k}{\sqrt{n}}-\frac{1}{2\sqrt{n}}, \frac{l}{\sqrt{n}}-\frac{1}{2\sqrt{n}}\right),$ where $1\le k,l\le \sqrt{n},\,\, $   $i=k\sqrt{n}+l.$
Then
$$\mathbb{E}\left({\min_{\pi}\sum_{i=1}^{n}d\left(X_i,Z_{\pi(i)}\right)}\right)=\Theta\left(\sqrt{\ln(n) n}\right), $$
where the infimum is over all permutations of $\{1,2,\dots,n\}$ and where $d$ is the Euclidean distance.
\end{theorem}

We are now ready \textbf{to extend} Theorem \ref{thm:tal} to the displacement to the power $a$ provided that $a>1.$
\begin{theorem}  
\label{thm:tal_a}
Fix $a> 1.$ 
Let $n=q^2$ for some $q\in\mathbb{N}.$ Assume that $n$  mobile sensors $X_1,X_2,\dots,X_n$ are thrown uniformly and independently at random on the unit square $[0,1]^2.$ 
Consider the non-random points $(Z_i)_{i\le n}$ evenly distributed as follows:\\
$Z_i=\left(\frac{k}{\sqrt{n}}-\frac{1}{2\sqrt{n}}, \frac{l}{\sqrt{n}}-\frac{1}{2\sqrt{n}}\right),$ where $1\le k,l\le \sqrt{n},\,\, $   $i=k\sqrt{n}+l.$
Then
$$
\mathbb{E}\left({\min_{\pi}\sum_{i=1}^{n}d^a\left(X_i,Z_{\pi(i)}\right)}\right)=\Omega\left((\ln(n))^{\frac{a}{2}}n^{1-\frac{a}{2}}\right),
$$
where the infimum is over all permutations of $\{1,2,\dots,n\}$ and where $d$ is the Euclidean distance.
\end{theorem}
\begin{proof} (Theorem \ref{thm:tal_a})
Let $\pi^{\star}\in S_n$ be a permutation with
$$
T^{(b)}=\sum_{i=1}^{n}d^b\left(X_i,Z_{\pi^{\star}(i)}\right)=\inf_{\pi\in S_n}\sum_{i=1}^{n}d^b\left(X_i,Z_{\pi(i)}\right),\,\,\,1\le b <\infty
$$
where $S_n$ is the set of all permutations of the numbers $1,2,\dots, n.$\\
Fix $a>1.$ Applying discrete H\"older inequality we get
$$
\sum_{i=1}^nd\left(X_i, Z_{\pi^{\star}(i)}\right)\le\left(\sum_{i=1}^n d^a\left(X_i, Z_{\pi^{\star}(i)}\right)\right)^{\frac{1}{a}}\left(\sum_{i=1}^n 1\right)^{\frac{a-1}{a}}.
$$
Hence
$$
\left(T^{(1)}\right)^a\le T^{(a)}n^{a-1}.
$$
Passing to the expectations and using  Jensen inequality for $X:=T^{(1)}$ and $f(x)=x^a$ 
we get the following estimation
\begin{equation}
\label{eq:lastese}
\left(\mathbb{E}\left({T^{(1)}}\right)\right)^a\le\mathbb{E}\left({T^{(a)}}n^{a-1}\right).
\end{equation}
Putting together Theorem \ref{thm:tal} and inequality (\ref{eq:lastese}) we obtain
$$
\mathbb{E}\left({T^{(a)}}\right)\ge n^{1-a}\left(\Theta\left(\sqrt{\ln(n) n}\right)\right)^a=\Theta\left((\ln(n))^{\frac{a}{2}}n^{1-\frac{a}{2}}\right).
$$
Therefore 
$$
\mathbb{E}\left({\inf_{\pi}\sum_{i=1}^{n}d^a\left(X_i,Z_{\pi(i)}\right)}\right)=\Omega\left((\ln(n))^{\frac{a}{2}}n^{1-\frac{a}{2}}\right).
$$
This completes the proof of Theorem \ref{thm:tal_a}.
\end{proof}
%
\section{Coverage \& interference requirement when the sensing radius $r_1>\frac{1}{2n}$ and the interference distance $s<\frac{1}{n}.$}
\label{sec:one}
In this section, we analyze the expected $a$-total displacement to fulfill
$(r_1,s)-C\&I$ \textit{requirement} when $n$ mobile sensors with identical sensing radius $r_1>\frac{1}{2n}$
are distributed uniformly at random and independently
on the unit interval $[0,1].$ That is, the sum of sensing area of $n$ sensors is \textbf{greater} than the \textbf{length of the unit interval.}
\subsection{Analysis of Algorithm \ref{alg_left}}
\begin{algorithm}
\caption{$MV(n,\rho, s)$ Moving sensors on $[0,1].$  } 
\label{alg_left}
\begin{algorithmic}[1]
 \REQUIRE The initial locations of $n$ mobile sensors, placed uniformly and independently at random on the unit interval $[0,1].$
 \ENSURE  The final positions of the sensors such that: 
 \begin{enumerate}
 \item[(i)] The distance between consecutive sensors 
is greater than or equal  to $s$ and less than or equal to $\rho.$ 
\item[(ii)]
The leftmost sensor is at a distance less than or equal to $\frac{\rho}{2}$ from the origin.
\end{enumerate}
\INITIALIZATION Sort the initial locations of $n$  sensors with respect to the origin of the interval, the location and sensors after sorting $X_{(1)}\le X_{(2)}\le\dots\le X_{(n)};$
  \STATE{Let $X_0=0;$}
 \FOR{$i=1$  \TO $n$ } 
 \IF{$X_{(i)}-X_{(i-1)}<s$}
 \STATE{move left to right the sensor $X_{(i)}$ to the new position $\min\left(s+X_{(i-1)},1\right);$}
 \ELSIF{$X_{(i)}-X_{(i-1)}>\rho$}
 \STATE{move right to left the sensor $X_{(i)}$ to the new position $\rho+X_{(i-1)};$}
 \ELSE
 \STATE{do nothing;}
 \ENDIF
 \ENDFOR
 \IF{$X_{(1)}>\frac{1}{2}\rho$}
 \STATE {$z:=X_{(1)}-\frac{1}{2}\rho;$}
 \FOR{$i=1$  \TO $n$ } 
 \STATE{move right to left the sensor $X_{(i)}$ to the new position $X_{(i)}-z;$}
 \ENDFOR
 \ENDIF
\end{algorithmic}
\end{algorithm}
Fix $a>0.$ Let $\gamma>0$, $1>\delta>0$ be arbitrary small constants independent on the number of sensors $n$ and let  $\rho= \frac{1+\gamma}{n},\,\,$ $s=\frac{1-\delta}{n}.$

This subsection is concerned with reallocating of the $n$ random sensors within the unit interval to achieve only the following property:
\begin{itemize}
\item The distance between consecutive sensors is greater than or equal  to $s$ and less than or equal to $\rho.$ 
\item The first leftmost sensor is at a distance less than or equal to $\frac{\rho}{2}$ from the origin.
\end{itemize}
We present basic and energy efficient algorithm $MV(n,\rho, s)$ (see Algorithm \ref{alg_left}). 
Theorem \ref{thm:covere} states that the expected $a$-total displacement
of algorithm $MV(n,\rho,s)$ 
is in
$
O\left(n^{1-a}\right)$
when
$\rho=\frac{1+\gamma}{n}$ and $s= \frac{1-\delta}{n}.$
Algorithm \ref{alg_left} is very simple but the asymptotic analysis is not totally trivial. We note that asymptotic analysis of Algorithm \ref{alg_left} is crucial
in deriving the threshold phenomena.

In the proof of Theorem \ref{thm:covere} we combine combinatorial techniques with properties of the Beta distribution 
(see Equation (\ref{eq:estaay}) in Lemma \ref{lemma_e} and Equation (\ref{eq:estaaybn}) in Lemma \ref{lemma_f}). 
The estimations for Beta distribution with special positive integers
parameters in Lemma \ref{lemma_e} and Lemma \ref{lemma_f} are new to the best of the author's knowledge.

Before starting the proof of Theorem \ref{thm:covere}, we briefly discuss one technical issue in the steps $(3)$-$(4)$ of Algorithm \ref{alg_left}. 
It may happen that for some initial random location of $n$ sensors $X_{(1)}\le X_{(2)}\le\dots \le X_{(n)}$ Algorithm \ref{alg_left} moves
some sensors to the right endpoint of the interval $[0,1].$
Namely, there exists $l_0\in\mathbb{N_{+}}$ with the following property
$X_{(i)}$ moves to some point in $[0,1)$  for all $i=1,2,\dots,l_0$ and $X_{(i)}$ moves to the right endpoint of the interval $[0,1]$ for all $i=l_0+1,l_0+2,\dots,n.$
 Let $Y_1,Y_2,\dots ,Y_n$ be the location of $n$ sensors $X_{(1)}\le X_{(2)}\le \dots\le X_{(n)}$ after Algorithm \ref{alg_left}.
Then to avoid interference to achieve the property that \textit{the distance between consecutive sensors is greater than or equal to $s$}, we have to deactivate some sensors.
Namely, 
\begin{itemize}
\item if $1-Y_{l_0}< s$ then  for all $i=l_0+1,l_0+2,\dots,n$ the sensors $X_{(i)}$ will not sense any longer,
 \item if $1-Y_{l_0}\ge s$ then for all $i=l_0+2,l_0+3,\dots,n$ the sensors $X_{(i)}$ will not sense any longer.
\end{itemize}

We are now ready to give the proof of Theorem \ref{thm:covere}.

\begin{theorem}
 \label{thm:covere}
Let $a> 0$ be a constant.  Fix $\gamma>0$, $1>\delta>0$ independent on the number of sensors $n.$ 
Assume that $n$ mobile sensors are thrown uniformly and independently at random on the unit interval $[0,1].$
Then Algorithm \ref{alg_left} for $\rho=\frac{1+\gamma}{n}$ and $s=\frac{1-\delta}{n}$ 
reallocates the random sensors within the unit interval so that:
\begin{itemize}
\item[(i)] The distance between consecutive sensors 
is greater than or equal  to $s$ and less than or equal to $\rho.$  
\item[(ii)] The leftmost sensor is at a distance less than or equal to $\frac{\rho}{2}$ from the origin.
\item[(iii)] The expected $a$-total displacement is $O\left(n^{1-a}\right).$ 
\end{itemize}
\end{theorem}
Notice that Theorem \ref{thm:covere} is valid regardless of the \textit{sensing radius}, it depends only on the fact that the relocated sensors are not too far.
\begin{proof}
Let  $\rho=\frac{1+\gamma}{n}$ and $s=\frac{1-\delta}{n},$ provided that $\gamma>0$, $1>\delta>0$ are arbitrary small constants independent on the number of sensors $n.$
Notice that Algorithm (\ref{alg_left}) is in two phases. During the first phase (see steps ($(1)$-$(10)$)) we reallocate the sensors so that the 
distance between consecutive sensors is greater than or equal  to $s$ and less than or equal to $\rho.$  In the second phase (see steps ($(11)$-$(16)$))
we reallocate the sensors to achieve the additional property the first leftmost sensor is in the distance less than or equal to $\frac{\rho}{2}$ from the origin.

Hence the properties (i) and (ii) hold and thus Algorithm 1 is \textbf{correct.}

We now estimate the expected $a$-total displacement of the algorithm.
 \\ 
 
\noindent
\textbf{First Phase} $\,\,\,\,$ The steps $(1)$-$(10)$ of Algorithm \ref{alg_left}\\
The main idea of the proof is simple.
Algorithm  \ref{alg_left} produces a sequence of moves for $X_{(i)}$ which consists of left moves (say $L$), right moves (say $R$) or no move at all (say $U$). 
Now, the idea of the proof is to chop the resulting set of moves into a run of $L$ followed by a run of $R$ followed by a run of $U$, etc. (Here runs might be empty as well.) Using this, 
we give an upper bound on the total displacement (namely the bound (\ref{eq:dusas})) whose expectation is then bounded. 
\\ 
Notice that exist $i\in\{1,2,\dots, n\}$ such that Algorithm \ref{alg_left} leaves the sensors $X_{(1)},
X_{(2)},\dots$ $X_{(i-1)}$  at the same positions.
(Here for $i=1$ Algorithm moves the sensor $X_{1}$)
Then the steps ($(1)$-$(10)$) of Algorithm \ref{alg_left} are the sequence of the two phases: $A$ and $B.$ During phase $A,$ Algorithm \ref{alg_left}
moves the sensors $X_{(i+1)},X_{(i+2)},$ $\dots X_{(i+p)}$ at the new positions. 
$k\in\{1,2,\dots, n\}.$ 
Then in phase $B,$ Algorithm \ref{alg_left}  leaves the sensors $X_{(i+p+1)},X_{(i+p+2)},$ $\dots X_{(i+p+k)}$
at the same positions. (Here phase $B$ might not exist and  Algorithm \ref{alg_left} moves the sensors $X_{(i+1)},X_{(i+2)},$ $\dots X_{(n)}$).

To better illustrate analysis, let us consider the following example.
Consider the phase $A$ as specified above. Let $p=p_1+p_2$ for some $p_1, p_2\in\mathbb{N_+}.$
\begin{enumerate}
 \item The sensors $X_{(i+1)}, X_{(i+2)},\dots X_{(i+p_1)}$ move right to left. 
 Observe that the sensors $X_{(i+1)},X_{(i+2)},\dots X_{(i+p_1)}$ have to move cumulatively, namely
 for $\ell=1,2,\dots, p_1$ the sensor $X_{(i+\ell)}$ moves right to left to the position $X_{(i)}+\rho \ell.$
 The  displacement to the power $a$ is 
 $$T^a_1=\sum_{\ell=1}^{p_1}\left(\left|X_{(i+\ell)}-X_{(i)}-\rho \ell\right|^{+}\right)^a.$$
 \item  The sensors $X_{(i+p_1+1)}, X_{(i+p_1+2)},\dots X_{(i+p_1+p_2)}$ move left to right. 
 Notice that the sensors $X_{(i+p_1+1)}, X_{(i+p_1+2)},\dots X_{(i+p_1+p_2)}$ have to move cumulatively, namely
 for $\ell=1,2,\dots, p_2$ the sensors $X_{(i+p_1+\ell)}$ move left to right to  the position $X_{(i)}+\rho p_1+s\ell.$
 The displacement to the power $a$ is\\ $T^a_2=\sum_{\ell=1}^{p_2}\left(\left|X_{(i)}+\rho p_1+s\ell-X_{(i+p_1+\ell)}\right|^{+}\right)^a.$
 Since $X_{(i)}+\rho p_1<X_{(i+p_1)}$ (see Figure \ref{fig:aser}) 
 we upper bound the displacement to the power $a$ as follows:
 $$T^a_2\le \sum_{\ell=1}^{p_2}\left(\left|X_{(i+p_1)}+s\ell-X_{(i+p_1+\ell)}\right|^{+}\right)^a.$$
\end{enumerate}
\begin{figure}[H]
\setlength{\unitlength}{0.75mm}
\centering
\begin{picture}(60,20)
\setlength{\unitlength}{0.75mm}
\put(25,20){\vector(1,0){33}}
\put(-28,10){\line(1,0){115}}
\put(-24,10){\circle*{2}}
\put(-28.5,5){$X_{(i)}+\rho p_1$}
\put(7,10){\circle*{2}}
\put(2.5,5){$X_{(i+p_1)}$}
\put(26,10){\circle*{2}}
\put(21.5,5){$X_{(i+p_1+\ell)}$}
\put(58,10){\circle*{2}}
\put(53.5,5){$X_{(i)}+\rho p_1+sl$}
\put(7,20){\vector(-1,0){32.5}}
\end{picture}
\caption{The movement of mobile sensors $X_{(i+p_1)},$  $X_{(i+p_1+\ell)}$ specified by  $1.$ and $2.$ in the phase $A$ of Algorithm \ref{alg_left}.}
\label{fig:aser}
\end{figure}

We are now ready to
estimate 
the movement of sensors in the 
phase $A$ in Algorithm \ref{alg_left}.
Let $p=p_1+p_2+\dots p_m$ for some $p_1, p_2,\dots p_m\in\mathbb{N_+}$ and $p_0=0.$ We assume that phase $A$ is divided into $m$ subphases as follows.
Algorithm \ref{alg_left} moves cumulatively the sensors $X_{(i+p_1+p_2+\dots p_{j-1}+1)},$ $X_{(i+p_1+p_2+\dots p_{j-1}+2)},\dots,$\\ $X_{(i+p_1+p_2+\dots p_{j-1}+p_j)}$
into one chosen direction left to right or right to left.
The movement direction of the sensors $X_{(i+p_1+p_2+\dots p_{j-1}+1)},$ $X_{(i+p_1+p_2+\dots p_{j-1}+2)},\dots,$ $X_{(i+p_1+p_2+\dots p_{j-1}+p_j)}$
is opposite to the movement direction of the sensors\\ 
$X_{(i+p_1+p_2+\dots p_{j}+1)},$ $X_{(i+p_1+p_2+\dots p_{j}+2)},\dots,$ $X_{(i+p_1+p_2+\dots p_{j}+p_{j+1})},$ 
provided that $j=1,2,\dots,m-1.$

Let $T^{a}_p$ be the displacement to the power $a$ in the considered phase $A$ of Algorithm \ref{alg_left} and let $p_0=0.$ Observe that
\begin{align}
\nonumber T^{a}_{p}&\le  \max_{0< p_1+\dots+ p_m\le p}\sum_{j=1}^{m}\sum_{\ell=1}^{p_j}\left( |X_{(i+p_1+\dots+p_{j-1}+\ell)}-X_{(i+p_1+\dots+p_{j-1})}-\rho\ell|^{+}\right)^a\\
\label{eq:tab}&+ \max_{0< p_1+\dots+ p_m\le p}\sum_{j=1}^{m}\sum_{\ell=1}^{p_j}\left( |X_{(i+p_1+\dots+p_{j-1})}+s\ell-X_{(i+p_1+\dots+p_{j-1}+\ell)}|^{+}\right)^a.
\end{align}
Let $T^a$  be the displacement to the power $a$ of Algorithm \ref{alg_left} in steps $(1)$-$(10)$ . Using  (\ref{eq:tab}), as well as the observation that Algorithm \ref{alg_left}
is the sequence of the two phases $A$ and $B$
we get the following upper bound
\begin{align}
&\nonumber T^{a}_{p}\le\\
\nonumber&  \max_{0\le p_1+\dots+ p_m\le p,\,\,\, 1\le p\le n}\sum_{j=1}^{m}\sum_{\ell=1}^{p_j}\left( |X_{(i+p_1+\dots+p_{j-1}+\ell)}-X_{(i+p_1+\dots+p_{j-1})}-\rho\ell|^{+}\right)^a\\
\label{eq:dusas}&+ \max_{0\le p_1+\dots+ p_m\le p,\,\,\, 1\le p\le n}\sum_{j=1}^{m}\sum_{\ell=1}^{p_j}\left( |X_{(i+p_1+\dots+p_{j-1})}+s\ell-X_{(i+p_1+\dots+p_{j-1}+\ell)}|^{+}\right)^a.
\end{align}
Let $b_{1,\ell}, b_{2,\ell}$ be some integers such that 
$b_{1,\ell}, b_{2,\ell}\in\{1,2,\dots, n\}$ and  $b_{1,l}-b_{2,l}=\ell.$ 
Observe that the following costs $\left(|X_{(b_{1,\ell})}-X_{(b_{2,\ell})}-\rho \ell|^{+}\right)^a$\\ and $\left( |X_{(b_{2,\ell})}+sl-X_{(b_{1,\ell})}|^{+}\right)^a$
can appear in the double sums (\ref{eq:dusas}) at most $\frac{n}{\ell}$ times. 
Hence
\begin{equation}
T^a\le \sum_{\ell=1}^{n}\frac{n}{\ell} \left(|X_{(b_{1,\ell})}-X_{(b_{2,\ell})}-\rho \ell|^{+}\right)^a
\label{eq:dusza}+\sum_{\ell=1}^{n}\frac{n}{\ell} \left(|s\ell-X_{(b_{1,l})}-X_{(b_{2,l})}|^{+}\right)^a.
\end{equation}
Let as recall the following claim.
\begin{Claim}
The random variable 
\begin{equation}
\label{eq:betabeta100}
X_{(j+\ell)}-X_{(j)}\,\,\,  \text{has the}\,\,\, \mathrm{Beta}(\ell,n-\ell+1)\,\,\, \text{distribution} 
\end{equation}
(see \cite[Formula 2.5.21, page 33]{arnold2008}).
\end{Claim}
Combining (\ref{eq:dusza}), (\ref{eq:betabeta100})  we have for the expectation value
\begin{align*}
\mathbb{E}\left(T^a\right)&\le \sum_{l=1}^{n}\frac{n}{l} \mathbb{E}\left(|\mathrm{Beta}(\ell,n-\ell+1)-\rho l|^{+}\right)^a\\
&+\sum_{l=1}^{n}\frac{n}{l} \mathbb{E}\left(|sl-\mathrm{Beta}(\ell,n-\ell+1)|^{+}\right)^a.
\end{align*}
Combining  Equation (\ref{eq:estaay}) in Lemma \ref{lemma_e} 
and
Equation (\ref{eq:estaaybn}) in Lemma \ref{lemma_f}  
lead to
$\mathbb{E}\left(T^a\right)=O\left(n^{1-a}\right).$ 
This is enough to prove the desired upper bound in the First Phase.
 \\ 
 
\noindent 
\textbf{Second Phase} $\,\,\,\,$ The steps $(11)$-$(16)$ of Algorithm \ref{alg_left}
\\

Observe that, after the steps $(1)$-$(10)$ the sensor $X_{(1)}$ has to be at position $P_1$ such that $0\le P_1\le \rho=\frac{1+\gamma}{n}.$
Hence for each sensor we upper bound the movement to the power $a$ by $\left(\frac{\rho}{2}\right)^a.$
Therefore, the expected $a$-total displacement of Algorithm \ref{alg_left} is less than
$$
\sum_{i=1}^{n}\left(\frac{\rho}{2}\right)^a=\frac{(1+\gamma)^a}{2^a}n^{1-a}=O\left(n^{1-a}\right).
$$
This is enough to prove the desired upper bound in the second case.
\\
Finally, combining together the estimation from \textbf{both phases} and Lemma \ref{lem:cruciana}
completes the proof of Theorem \ref{thm:covere}.
\end{proof}
Finally, the following lemma will be helpful in the proof of the main results  in Subsection \ref{subsec:astra} for the sensors on the unit interval.
In the proof of Lemma \ref{lem_claim} we combine probabilistic techniques together with Estimation (\ref{eq:estaax}) in Lemma \ref{lemma_e} for Beta distribution from 
Section \ref{sec:preliminary}.
\begin{Lemma}
\label{lem_claim}
Let $a>0$ be a constant. Fix $\gamma>0$, $1>\delta>0$ independent on the number of sensors $n.$ Let  $\rho=\frac{1+\gamma}{n}$ and $s=\frac{1-\delta}{n}.$
Let $Y_n$ be the location of $n$-th sensor after algorithm $MV(n,\rho, s).$ Then
$$\Pr\left[Y_n< 1-\frac{2}{n^{\frac{a}{a+1}}}\right]=O\left(\frac{1}{n^{\frac{a}{2}}}\right).$$
\end{Lemma}
\begin{proof}
Let $M_{n}(1-10)$ be the movement of sensor $X_{(n)}$ right to left in Algorithm \ref{alg_left} at steps $(1)$-$(10).$
The analysis of $M_{n}(1-10)$ is analogous to that in the proof of
Theorem \ref{thm:covere}. 
Using Equation (\ref{eq:estaax}) in Lemma \ref{lemma_e} for ${\frac{(a+1)a}{2}}$ we get
\begin{equation}
\label{eq:markkova}
\mathbb{E}\left[\left(M_n(1-10)\right)^{\frac{(a+1)a}{2}}\right]=O\left(\frac{1}{n^{\frac{(a+1)a}{2}}}\right).
\end{equation}
Let $M_n(11-16)$ be the movement of sensor $X_{(n)}$ right to left in Algorithm \ref{alg_left} at the steps $(11)$-$(16).$
Observe that $M_n(11-16)\le\frac{\rho}{2}=\frac{1}{2}\frac{1+\gamma}{n}.$ Therefore
\begin{equation}
\label{eq:markkovb}
\mathbb{E}\left[\left(M_n(11-16)\right)^{\frac{(a+1)a}{2}}\right]=O\left(\frac{1}{n^{\frac{(a+1)a}{2}}}\right).
\end{equation}
Let $M_n$ be the movement of sensor $X_{(n)}$ right to left in Algorithm \ref{alg_left}. Putting together the equality
$M_n=M_n(1-10)+M_n(11-16),$ Estimations (\ref{eq:markkova}-\ref{eq:markkovb}), as well as Lemma \ref{lem:cruciana} we have
\begin{equation}
\label{eq:markkov}
\mathbb{E}\left[\left(M_n\right)^{\frac{(a+1)a}{2}}\right]=O\left(\frac{1}{n^{\frac{(a+1)a}{2}}}\right).
\end{equation}
Applying Markov inequality applied for random variable $M_n^{\frac{(a+1)a}{2}}$
and Estimation (\ref{eq:markkov}) we deduce that 
\begin{equation}
\label{eq:markkoves}
\Pr\left[M_n>\frac{1}{n^{\frac{a}{1+a}}} \right]=\Pr\left[\left(M_n\right)^{\frac{(a+1)a}{2}}>\frac{1}{n^{\frac{a^2}{2}}} \right]= O\left(\frac{n^{\frac{a^2}{2}}}{n^{\frac{(a+1)a}{2}}}\right)
=O\left( \frac{1}{n^{\frac{a}{2}}}\right).
\end{equation}
Consider the following three events:
$$
E_1: Y_n<1-2n^{-\frac{a}{a+1}}\,\, |\,\, X_{(n)}\ge 1-n^{-\frac{a}{a+1}},
$$
$$
E_2: Y_n<1-2n^{-\frac{a}{a+1}}\,\, |\,\, X_{(n)}< 1-n^{-\frac{a}{a+1}},
$$
$$
E_3: X_{(n)}< 1-n^{-\frac{a}{a+1}}.
$$
Applying Equation (\ref{eq:markkoves}) yields
$$
\Pr\left[E_1\right]\left(1-\Pr\left[E_3\right]\right)\le \Pr\left[E_1\right]\le\Pr\left[M_n>\frac{1}{n^{\frac{a}{a+1}}} \right]=O\left( \frac{1}{n^{\frac{a}{2}}}\right).
$$
From Lemma \ref{lemma_first}, as well as the fact that random $X_{(n)}$ obeys $\mathrm{Beta}(n,1)$ we have
$$
\Pr\left[E_2\right]\Pr\left[E_3\right]\le\Pr\left[E_3\right]<\frac{1}{e^{n^{\frac{1}{1+a}}}}\,\,\,\,\,\,\text{is exponentially small.}
$$
Putting all together we deduce that
$$
\Pr\left[Y_n<1-\frac{2}{n^{\frac{a}{a+1}}} \right]=\Pr\left[E_1\right]\left(1-\Pr\left[E_3\right]\right)+\Pr\left[E_2\right]\Pr\left[E_3\right]
=O\left( \frac{1}{n^{\frac{a}{2}}}\right).
$$
This finishes the proof of Lemma  \ref{lem_claim}. 
\end{proof}
\subsection{Analysis of Algorithm \ref{alg_coverage}}
\label{subsec:astra}
Let us recall that $a>0$ is fixed and $\epsilon>0$, $1>\delta>0$ are arbitrary small constants independent on the number of sensors $n.$ 
In this subsection we present algorithm
$CV_1(n,r_1,s)$ (see Algorithm \ref{alg_coverage}) 
for the $(r_1,s)-C\&I$ requirement.
We prove that the expected $a$-total
displacement of algorithm $CV_1(n,r,s)$ 
is in
$O\left(n^{1-a}\right)$
when
$r_1= \frac{1+\epsilon}{2n}$
and
$s=\frac{1-\delta}{n}.$
Notice that our Algorithm \ref{alg_coverage} consists of  two phases. During the first phase (see Initialization) we apply Algorithm \ref{alg_left}.
Then in the second phase (see Case \textbf{B} and Case \textbf{C}) we add the additional sensors movement. 
Let $Y_n$ be the location of sensors $X_{(n)}$ after Algorithm \ref{alg_coverage}.
The additional movement depends on the position of sensor $Y_{n}$ in the interval $[0,1].$
\begin{algorithm}[H]
\caption{$CV_1(n,r_1,s)$ for $(r_1,s)$-\textit{coverage \& interference requirement} on $[0,1]$ when  $r_1=\frac{1+\epsilon}{2n},\,\,\,$ $s=\frac{1-\delta}{n}$ provided that $\epsilon>0$, $1>\delta>0$ are fixed 
and independent on $n.$}
\label{alg_coverage}
\begin{algorithmic}[1]
 \REQUIRE The initial locations of $n$ mobile sensors with identical sensing radius $r_1=\frac{1+\epsilon}{2n}$, placed uniformly and independently at random on the unit interval $[0,1].$
 \ENSURE  The final positions of sensors to satisfy $(r_1,s)$-\textit{coverage \& interference  requirement} on the interval $[0,1].$
 
 \INITIALIZATION Apply Algorithm $MV(n,\rho,s)$ for  $\rho:=\frac{1+\frac{\epsilon}{2}}{n},\,\,\,$ $s:=\frac{1-\delta}{n}$ and the random sensors $X_1,X_2,\dots,X_n.$
 Let $Y_1,Y_2,\dots ,Y_n$ be the location and sensors of $n$ sensors $X_{(1)}\le X_{(2)}\le \dots\le X_{(n)}$ after Algorithm $MV(n,\rho, s);$
 \SWITCH {}
  \CASE {\textbf{A} $\left(Y_n\ge 1-r_1\right)$}
    \STATE {do nothing;}
  \ENDCASE
  \CASE {\textbf{B} $\left(Y_n\le 1-\frac{2}{n^{\frac{a}{a+1}}}\right)$}
   \FOR{$i=1$  \TO $n$ } 
 \STATE{move the sensor $Y_i$ to the position $\left(\frac{i}{n}-\frac{1}{2n}\right);$}
 \ENDFOR
  \ENDCASE
  \CASE {\textbf{C} $\left( Y_n\in\left(1-\frac{2}{n^{\frac{a}{a+1}}}, 1-r_1\right)\right)$}
  \STATE {move the sensor $Y_n$ to the new position $1-r_1,\,\,\,$ $i:=n-1;$}
   \WHILE{$Y_{i+1}-Y_{i}>2r_1 $}
    \STATE {move the sensor $Y_{i}$ to the new position $1-r_1-(n-i)2r_1,\,\,\,$ $i:=i-1$;}
  \ENDWHILE
  \ENDCASE
\ENDSWITCH
\end{algorithmic}
\end{algorithm}
We now briefly explain \textbf{the ideas} behind the proof of  Theorem \ref{thm:const} 
and \textbf{correctness} of Algorithm \ref{alg_coverage}. 
\begin{itemize}
\item[(i)] We have initially $n$ random sensors on the unit interval with identical sensing radius $r_1=\frac{1+\epsilon}{2n}.$
Firstly, we apply Algorithm \ref{alg_left} for  $\rho=\frac{1+\frac{\epsilon}{2}}{n}$ and $s=\frac{1-\delta}{n}$ to achieve only the following property:
\begin{itemize}
\item The distance between consecutive sensors is greater than or equal  to $\frac{1-\delta}{n}$ and less than or equal to $\frac{1+\frac{\epsilon}{2}}{n}.$ 
\item The first leftmost sensor is at a  distance less than or equal to $\frac{1+\frac{\epsilon}{2}}{2n}$ from the origin.
\end{itemize}
Applying Theorem \ref{thm:covere} we deduce that the 
expected $a$-total displacement in \textbf{Initialization} of  Algorithm  \ref{alg_coverage} is $O\left(n^{1-a}\right).$
\item[(ii)] In case B we move the sensors to equidistant anchor locations in $\Theta\left(n^{1-\frac{a}{2}}\right)$ expected $a$-total displacement.
However, we can upper bound the probability with which case B occurs (see
Lemma \ref{lem_claim}) to achieve the desired
$O\left(n^{1-a}\right)$ expected $a$-total displacement.
\item[(iii)] Since the sensors have sensing radius $r_1=\frac{1+\epsilon}{2n}$ and the distance between consecutive sensors is less than or equal to 
$\frac{1+\frac{\epsilon}{2}}{n}=2r_1-\frac{\frac{\epsilon}{2}}{n},$
$(r_1,s)$-\textit{coverage \& interference requirement} is solved
in $O\left(n^{1-a}\right)$ expected $a$-total displacement in case C of Algorithm \ref{alg_coverage}. In this case only fraction
$\Theta\left(n^{\frac{1}{a+1}}\right)$ of rightmost sensors can move. We upper bound the movement to the power $a$ of each these sensors by
$\frac{2^a}{n^{\frac{a^2}{a+1}}}$ (see Case 3 in the proof of Theorem \ref{thm:const}).
\end{itemize}

We are now ready to prove the main theorem for the sensors on the unit interval.
\begin{theorem}
\label{thm:const}
Let $a>0$ be a constant. Fix $\epsilon>0$, $1>\delta>0$ independent on the number of sensors $n.$  Let $s=\frac{1-\delta}{n}.$
Assume that $n$ mobile sensors with identical sensing radius $r_1=\frac{1+\epsilon}{2n}$ are thrown uniformly and independently at random on the unit interval $[0,1].$
Then Algorithm \ref{alg_coverage} 
solves $(r_1,s)$-\textit{coverage \& interference requirement} and has expected $a$-total displacement  
$O\left(n^{1-a}\right).$
\end{theorem}
\begin{proof}
There are three cases to consider.
\\

{\textbf{Case $1$:}} The algorithm terminates after Step $3$. This case adds nothing to the expected $a$-total displacement.  
\\

{\textbf{Case $2$:}} The algorithm terminates after Step $7$. Then  $Y_n\le 1-\frac{2}{n^{\frac{a}{a+1}}}.$

In this case we upper bound the expected $a$-total displacement in steps $(5)$-$(7)$ of algorithm $CV_1(n,r_1,s)$ as follows:
\begin{itemize}
\item[(a)] rewind the $i$-th sensor from the location $Y_i$ to the location $X_{(i)}$ for $i=1,2,\dots ,n.$ From Theorem \ref{thm:covere} we get back
the expected $a$-total displacement is $O\left(n^{1-a}\right).$
\item[(b)] Move the $i$-th sensor from the location $X_i$ to the position $\left(\frac{i}{n}-\frac{1}{2n}\right)$ for $i=1,2,\dots , n.$
According to Theorem \ref{thm:exact_one} the expected $a$-total displacement is\\ $\Theta\left(n^{1-\frac{a}{2}}\right).$
\end{itemize}
Putting together (a), (b), as well as Lemma \ref{lem:cruciana} we have the expected $a$-total displacement at the steps $(5)$-$(7)$ of algorithm $CV_1(n,r_1,s)$ is $O\left(n^{1-\frac{a}{2}}\right).$
Then by Lemma \ref{lem_claim} the probability that this case can occur is $O\left(\frac{1}{n^{\frac{a}{2}}}\right)$
and this adds to the expected $a$-total displacement at most
$$O\left(n^{1-\frac{a}{2}}\right)O\left(\frac{1}{n^{\frac{a}{2}}}\right)=O\left(n^{1-a}\right).$$
\\

{\textbf{Case $3$:}} The algorithm terminates after Step $12$. Then $Y_n\in\left(1-\frac{2}{n^{\frac{a}{a+1}}}, 1-r\right).$

Let us recall that $r_1=\frac{1+\epsilon}{2n},\,\,\,\rho=\frac{1+\frac{\epsilon}{2}}{n}$
and the distance between consecutive sensors is less than or equal to $\rho.$
Hence, we upper bound the movement to the power $a$ of the $(n-i)$-th sensor for $i\ge 1$ as follows:
\begin{align*}
&\left(\left|1-r_1-(n-i)2r_1-\left(1-\frac{2}{n^{\frac{a}{a+1}}}-\rho(n- i)\right)\right|^{+}\right)^a\\
&=\left(\left|\frac{2}{n^{\frac{a}{a+1}}}-\frac{\epsilon(n-i)+1+\epsilon}{2n}\right|^{+}\right)^a\le \frac{2^a}{n^{\frac{a^2}{a+1}}}.
\end{align*}
Observe that the movement of $(n-i)$-th sensor is positive only when
$$n-i\le \frac{4n^{{\frac{1}{a+1}}}}{\epsilon}-\frac{1}{\epsilon}=\Theta(n^{{\frac{1}{a+1}}}).$$
From this, we see that only  $\Theta\left(n^{{\frac{1}{a+1}}}\right)$
sensors can move. 

Observe that the movement to the power $a$ of the $n$-th sensor is also less then $\frac{2^a}{n^{\frac{a^2}{a+1}}}.$

Hence, this adds to the $a$-total displacement 
$$\frac{2^a}{n^{\frac{a^2}{a+1}}}\left(\Theta\left(n^{{\frac{1}{a+1}}}\right)+1\right)
=O\left(n^{1-a}\right).$$

Finally, combining together the estimation from Initialization (see Theorem \ref{thm:covere}), Case $1$, Case $2$, Case $3$, as well as Lemma \ref{lem:cruciana}
we conclude that the expected $a$-total displacement of algorithm $CV_1(n,s,r)$ is at most
$O\left(n^{1-a}\right).$
This is enough to prove Theorem \ref{thm:const}. 
\end{proof}
\section{Coverage \& interference requirement for square sensing radius $r_2>\frac{1}{2\sqrt{n}}$ and interference distance $s<\frac{1}{\sqrt{n}}$}
\label{sec:two}
In this section, we analyze the expected $a$-total displacement to achieve
$(r_2,s)-C\&I$ \textit{requirement} when $n$ mobile sensors with identical square sensing radius $r_2>\frac{1}{2\sqrt{n}}$
are thrown uniformly at random and independently on the unit square $[0,1]^2,$ 
That is, the sum of sensing area of $n$ sensors is \textbf{greater}  than the \textbf{area of unit square.}

Let us recall that $a>0$ is constant and $\epsilon, \delta>0$ are fixed arbitrary small constant independent on the number of sensors $n.$

We prove that the expected $a$-total expected displacement of
algorithm\\ $CV_2(n,r_2, s)$ (see Algorithm \ref{alg_coverage2}) 
is in
$O\left(n^{1-\frac{a}{2}}\right)$ when $r_2= \frac{1+\epsilon}{2\lfloor \sqrt{n}\rfloor}$ and $s=\frac{1-\delta}{\lfloor\sqrt{n}\rfloor}.$
\begin{algorithm}
\caption{$CV_2(n,r_2,s)$ for $(r_2,s)$-\textit{coverage \& interference requirement} on the $[0,1]^2$ when $r_2=\frac{1+\epsilon}{2\lfloor \sqrt{n}\rfloor}$ and $s=\frac{1-\delta}{\lfloor \sqrt{n}\rfloor}$ 
provided that $\epsilon>0$, $1>\delta>0$ are fixed 
and independent on $n.$}
\label{alg_coverage2}
\begin{algorithmic}[1]
 \REQUIRE The initial locations of $n$ mobile sensors with identical square sensing radius $r_2=\frac{1+\epsilon}{2\lfloor \sqrt{n}\rfloor}$, placed uniformly and independently at random on the unit square $[0,1]^2.$
 \ENSURE  The final positions of sensors to satisfying $(r_2,s)$-\textit{coverage \& interference  requirement} on the square $[0,1]^2.$
 
\INITIALIZATION  \begin{itemize} \item Choose $\lfloor \sqrt{n}\rfloor^2$ sensors at random;
 \item Sort the initial locations of sensors according to the second coordinate; let the sorted locations be
   $S_1=(x_1,y_1),$ $S_2=(x_2,y_2),\dots$  $S_n=(x_n,y_n),\,\,\,$ $y_1\le y_2\le \dots \le y_n;$ 
   \end{itemize}
   \FOR{$j=1$  \TO $\lfloor \sqrt{n}\rfloor$ } 
 \FOR{$i=1$  \TO $\lfloor \sqrt{n}\rfloor$ } 
 \STATE{move sensor $S_{(j-1)\lfloor \sqrt{n}\rfloor+i}$ to position\\ $\left(x_{(j-1)\lfloor \sqrt{n}\rfloor+i}, \frac{j}{\lfloor\sqrt{n}\rfloor}-\frac{1}{2\lfloor\sqrt{n}\rfloor}\right)$}
 \ENDFOR  
 \ENDFOR
  \FOR{$j=1$ to $\lfloor\sqrt{n}\rfloor$}
  \STATE{Apply Algorithm $CV_1(n,r_1,s)$ for $n:=\lfloor\sqrt{n}\rfloor,$ $s:=\frac{1-\delta}{\lfloor\sqrt{n}\rfloor},$ $r_1:=\frac{1+{\epsilon}}{2\lfloor\sqrt{n}\rfloor}$ and sensors 
  $S_{(j-1)\lfloor\sqrt{n}\rfloor+1}, S_{(j-1)\lfloor\sqrt{n}\rfloor+2},\dots S_{(j-1)\lfloor\sqrt{n}\rfloor+\lfloor\sqrt{n}\rfloor}; $ }
  \ENDFOR
\end{algorithmic}
\end{algorithm}
Notice that our Algorithm \ref{alg_coverage2} is in two phases. 
During the first phase (see steps $(1)$-$(7)$) we use a greedy strategy and move all the sensors only according to second coordinate.
As a result of the first phase we get $\lfloor\sqrt{n}\rfloor$ lines each with $\lfloor\sqrt{n}\rfloor$ random sensors.
For the second phase the main result from Section \ref{sec:one} (see Theorem \ref{thm:const}) is applicable. 

It is worth pointing out that the first phase of Algorithm \ref{alg_coverage2} reduces the $a$-total displacement on the unit square to the $a$-total displacement on the unit interval.
Obviously Algorithm \ref{alg_coverage2} moves sensors only in vertical and horizontal fashion but it is powerful enough to derive the desired threshold.

We are now ready to prove the main result for the sensor on the unit square.
\begin{theorem}
\label{thm:const2}
Let $a>0$ be a constant.
Fix $\epsilon>0$, $1>\delta>0$ arbitrary small constans independent on the number of sensors $n.$ 
Let $s=\frac{1-\delta}{\lfloor \sqrt{n}\rfloor}.$
Assume that $n$ mobile sensors with identical square
sensing radius $r_2=\frac{1+\epsilon}{2\lfloor\sqrt{n}\rfloor}$ are thrown uniformly and independently at random on the unit square $[0,1]^2.$
Then Algorithm \ref{alg_coverage2} 
solves $(r_2,s)$-\textit{coverage \& interference requirement} and has expected $a$-total displacement in $O\left(n^{1-\frac{a}{2}}\right).$
\end{theorem}
\begin{proof} (Theorem \ref{thm:const2})
Firstly, we look at the expected $a$-total displacement in first phase of the algorithm (see steps $(1)$-$(7)$).
It was proved in \cite{KK_2016_cube} that the expected $a$-total displacement in steps $(1)$-$(7)$ of Algorithm \ref{alg_coverage2} is in
$O\left(n^{1-\frac{a}{2}}\right)$ (see estimation of $E^{(a)}_{(1-6)}$ for $n:=\left(\lfloor \sqrt{n}\rfloor\right)^2,$ 
$d=2$ in the proof of  \cite[Theorem 5, Formulas (8), (10), page 41]{KK_2016_cube}).

Observe that in the second phase of Algorithm \ref{alg_coverage2} (see steps $(8)$-$(10)$) we have $\lfloor \sqrt{n}\rfloor$ lines each with $\lfloor\sqrt{n}\rfloor$ random sensors
with identical sensing radius $r_1=\frac{1+\epsilon}{2\lfloor \sqrt{n}\rfloor}.$
According to Theorem \ref{thm:const} 
the expected $a$-total displacement is $\lfloor \sqrt{n}\rfloor O\left(\left(\lfloor\sqrt{n}\rfloor\right)^{1-a}\right)=O\left(n^{1-\frac{a}{2}}\right).$
This together with Lemma \ref{lem:cruciana} completes the proof of Theorem \ref{thm:const2}. 
\end{proof}
\section{Sensors in higher dimensions}
\label{sec:discussion}
In this section we discuss the expected $a$-total displacement for $(r_m,s)$-\textit{coverage \& interference} requirement in higher dimensions, when $m>2.$

Let us recall that the proposed Algorithm \ref{alg_coverage2} 
moves the sensors only vertical and horizontal fashion and reduces the $a$-total displacement on the unit square to the $a$-total displacement on the unit interval.

Hence Algorithm \ref{alg_coverage2}  can be extended for the random sensors on the $m$-dimensional cube $[0,1]^m,$ when $m>2.$ 
We can similary to Square Sensing Radius (see Definition \ref{def:square}) define
$m$-Dimensional Cube Sensing Radius, move the sensors only according to the axes and reduce the $a$-total displacement on the unit cube to the $a$-total displacement on the unit interval.

Namely, for the sensors with the identical $m$-cube sensing radius 
$r_m>\frac{1}{2n^{1/m}}$ (the sum of sensing area of $n$ sensors is greater than the area of unit cube)
and the interference distance $s<\frac{1}{n^{1/m}}$ it is possible to give an algorithm  with $O\left(n^{1-\frac{a}{m}}\right)$ expected $a$-total displacement for all powers $a>0.$
However, even though Theorem \ref{thm:const2} can be generalized for the random sensors with the identical $m$-cube sensing radius 
$r_m>\frac{1}{2n^{1/m}}$ on $m$-dimensional cube, when $m>2,$ the \textit{proposed generalization is weak.}

Notice that Theorem \ref{thm:tal} is closely related to the main result of paper \cite{ajtai_84}.
Namely, consider two sequences $X_1, X_2,\dots, X_n;$ $Y_1, Y_2,\dots, Y_n$ of points that are
independently uniformly distributed 
and the non-random points $(Z_i)_{i\le n}$ are evenly distributed, i.e.
$Z_i=\left(\frac{k}{\sqrt{n}}-\frac{1}{2\sqrt{n}}, \frac{l}{\sqrt{n}}-\frac{1}{2\sqrt{n}}\right),$ where $1\le k,l\le \sqrt{n},\,\, $   $i=k\sqrt{n}+l.$
on the unit square $[0,1]^2$ then
$$\mathbb{E}\left({\inf_{\pi}\sum_{i=1}^{n}d\left(X_i,Z_{\pi(i)}\right)}\right)=\mathbb{E}\left({\inf_{\pi}\sum_{i=1}^{n}d\left(X_i,Y_{\pi(i)}\right)}\right)=\Theta\left(\sqrt{\ln(n) n}\right),$$
where $\pi$ ranges over all permutations of $\{1,2,\dots, n\}$ and $n=q^2$ for some $q\in\mathbb{N}.$

On the other hand, there is a difference between $m=2$ (the $2$-dimensional case) and $m>2$ (the case of dimension at least $3$). 
Namely for two sequences $X_1, X_2,\dots,$ $X_n;$ $Y_1, Y_2,\dots, Y_n$ of points that are
independently uniformly distributed 
on the $m$-dimensional cube $[0,1]^m,$ when $m>2$ we have
$$\mathbb{E}\left({\inf_{\pi}\sum_{i=1}^{n}d\left(X_i,Y_{\pi(i)}\right)}\right)=\Theta\left(n^{1-\frac{1}{m}}\right),$$
provided that $\pi$ ranges over all permutations of $\{1,2,\dots, n\}$
(see \cite{yukich} for details). 

Hence, it seems that Theorem \ref{thm:tal}
together with Theorem \ref{thm:tal_a} can be generalized for $n$ random mobile sensors $X_1,X_2,\dots,X_n$ on the $m$-dimensional cube $[0,1]^m,$ when $m>2$
and the following result should hold.

Assume that $n$ random variables $X_1, X_2,\dots, X_n$ are independently uniformly distributed 
and the 
non-random points $(Z_i)_{i\le n}$ evenly distributed
at the the positions
$
\left(\frac{l_1}{n^{1/d}}-\frac{1}{2n^{1/d}},
\frac{l_2}{n^{1/d}}-\frac{1}{2n^{1/d}},
\dots ,\frac{l_d}{n^{1/d}}-\frac{1}{2n^{1/d}}
\right),
$
for $1\le l_1,l_2,\dots , l_d\le n^{1/d}$ and $l_1,l_2,\dots,l_d\in \mathbb{N}$ 
on the unit $m$-dimensional cube $[0,1]^m$ then
\begin{equation}
\label{eq:open}
\mathbb{E}\left({\inf_{\pi}\sum_{i=1}^{n}d^a\left(X_i,Z_{\pi(i)}\right)}\right)=\Theta\left(n^{1-\frac{a}{m}}\right)
\end{equation}
for all powers $a\ge 1,$  
where $\pi$ ranges over all permutations of $\{1,2,\dots, n\}$ and
$n=q^m$ for some $q\in\mathbb{N}.$


Therefore, it is an open problem to prove that
$(r_m,s)$-\textit{coverage \& interference} requirement for  $m$-cube sensing radius 
$r_m=\frac{1}{2n^{1/m}}$ (the sum of sensing area of $n$ sensors is equal to the area of unit cube)
and the interference distance $s=\frac{1}{n^{1/m}}$ can be solved in $\Theta\left(n^{1-\frac{a}{m}}\right)$
and to study the expected $a$-total displacement for $(r_m,s)$-\textit{coverage \& interference} requirement, when 
$r_m>\frac{1}{2n^{1/m}}$ and $s<\frac{1}{n^{1/m}}.$
\section{Experimental Results}
\label{sec:simu}
In this section we provide a set of experiments to \textit{confirm} the discovered theoretical threshold
for the expected $a$-total displacement. Wolfram Mathematica $10.0$ was used for our experiments when $a=1,$ $a=\frac{3}{2}$ and $a=2.$
We distinguish two cases. \\

\textit{Case $1$: sensing radius $r_1>\frac{1}{2n}$ and interference distance $s<\frac{1}{n}.$}\\

In this case, we conduct Algorithm \ref{sym_Case1}.

Figures \ref{fig:alg2}, \ref{fig:alg32b} and \ref{fig:alg4} illustrates the described experiment for 
$a=1,$ $a=\frac{3}{2}$ and $a=2$

\begin{algorithm}[H]
\caption{}
\label{sym_Case1}
\begin{algorithmic}[1]
\STATE{$n:=1$}
\WHILE{n $\le$ 5000}
 \STATE{Generate independently and uniformly $n$ random points on the unit interval $[0,1]$;}
 \STATE{Calculate  $\mathbf{T}^{(a)}_{n}$ according to Algorithm \ref{alg_left} for $\rho=\frac{1.8}{n}$ and $s=\frac{0.5}{n};$}
 \STATE{Insert the points $(n, T^{(a)}_{n})$ into the chart;}
 \STATE{$n:=n+1$}
 \ENDWHILE
\end{algorithmic}
\end{algorithm} 

Notice that the experimental $a$-total displacement of Algorithm \ref{sym_Case1} is \textit{constant} and independent on the number of sensors for $a=1,$
is $\Theta\left(\frac{1}{\sqrt{n}}\right)$ for $a=\frac{3}{2}$
and is $\Theta\left(\frac{1}{n}\right)$ for $a=2.$
Therefore,  the carried out experiments confirm very well our theoretical upper bound estimation
$O(1)$ for $a=1,$
$O\left(\frac{1}{\sqrt{n}}\right)$ for $a=\frac{3}{2}$
and $O\left(\frac{1}{n}\right)$ for $a=2.$
(see Theorem \ref{thm:covere} for $a=1,$ $a=\frac{3}{2}$ and $a=2$).\\

\textit{Case $2$: sensing radius $r_1=\frac{1}{2n}$ and interference distance $s=\frac{1}{n}.$ }\\

In this case, we conduct Algorithm \ref{sym_Case2}.
\begin{algorithm}[H]
\caption{}
\label{sym_Case2}
\begin{algorithmic}[1]
\STATE{$n:=1$}
\WHILE{n $\le$ 60}
 \FOR{$j=1$  \TO $200$ } 
 \STATE{Generate independently and uniformly $n^2$ random points on the unit interval $[0,1]$;}
 \STATE{Calculate  $\mathbf{T}^{(a)}_{n^2}(j)$ according to Theorem \ref{thm:exact_one};}
 \ENDFOR
  \FOR{$k=1$  \TO $20$ } 
 \STATE{Calculate  the average $T^{(a)}_{n^2,k}=\frac{1}{10}\sum_{j=1}^{10} \mathbf{T}^{(a)}_{n^2}(j+(k-1)*10)$;}
 \STATE{Insert the points $(n^2, T^{(a)}_{n^2,k})$ into the chart;}
 \ENDFOR
 \STATE{$n:=n+1$}
 \ENDWHILE
\end{algorithmic}
\end{algorithm}
In Figures \ref{fig:alg1}, \ref{fig:alg32} and \ref{fig:alg3} the black points represents numerical results of conducted experiments. The additional lines
$\left\{\left(n, \frac{\Gamma\left(\frac{3}{2}\right)}{2\sqrt{2}}\sqrt{n} \right),\,\, 1\le n\le 3600\right\},$\\
$\left\{\left(n, \frac{\Gamma\left(\frac{7}{4}\right)}{2^{\frac{3}{4}}\left(\frac{5}{2}\right)}n^{\frac{1}{4}} \right),\,\, 1\le n\le 3600\right\},$
$\left\{\left(n, \frac{1}{6} \right),\,\, 1\le n\le 3600\right\}$ are the plots of function which is the theoretical estimation
(see the leading term in asymptotic result of Theorem \ref{thm:exact_one} for $a=1,$ $a=\frac{3}{2}$ and $a=2$).
It is worth pointing out that numerical results are situated near the theoretical line. 

It is also possible to repeat experiments to all 
exponents $a>0,$ as well as Algorithms \ref{alg_coverage} and \ref{alg_coverage2}.

 \begin{minipage}{\linewidth}
      \centering
      \begin{minipage}{0.45\linewidth}
          \begin{figure}[H]
              \includegraphics[width=\linewidth]{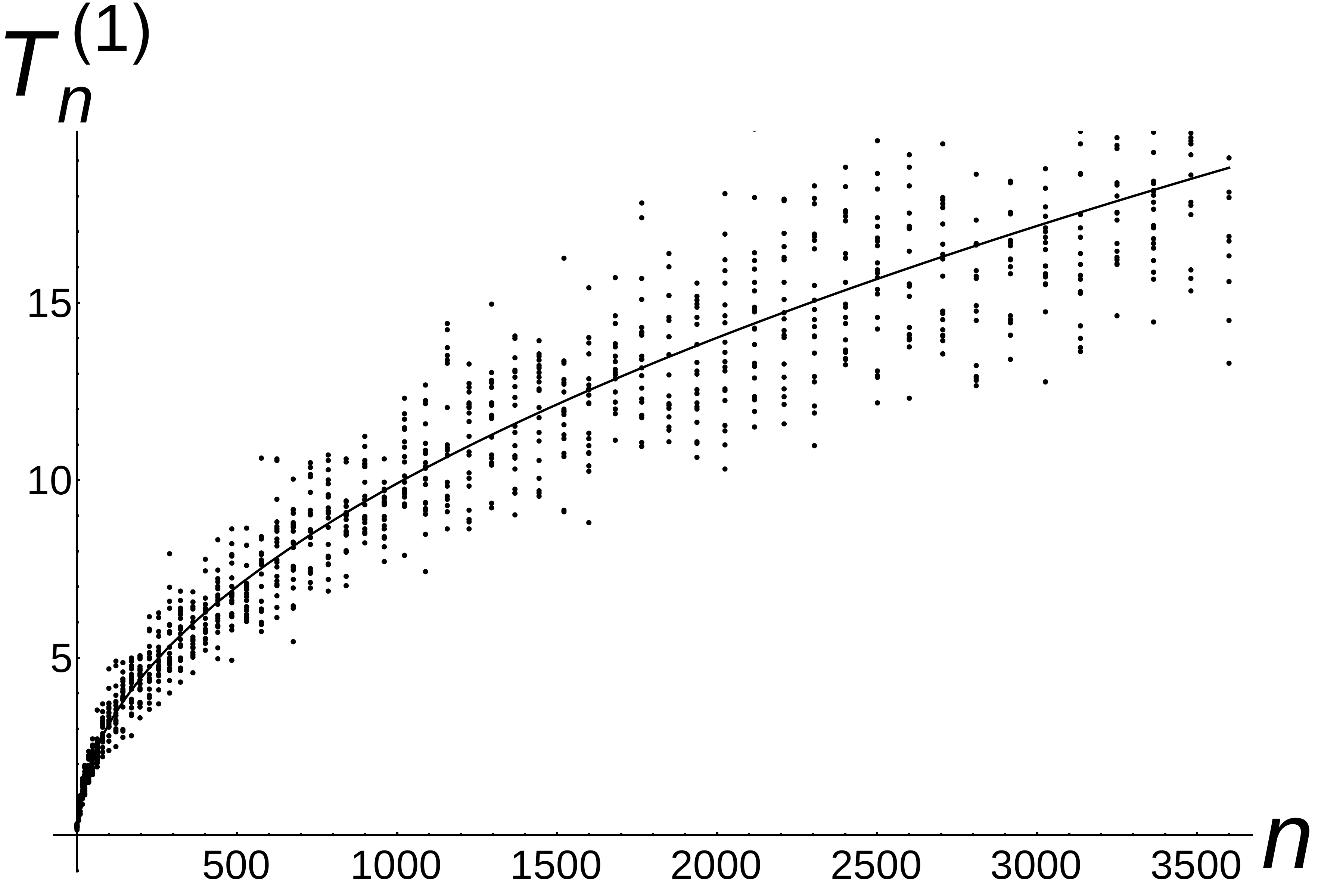}
              \caption{\label{fig:alg1} $T_n^{(1)}\sim \frac{\Gamma\left(\frac{3}{2}\right)}{2\sqrt{2}}\sqrt{n}$ of Algorithm \ref{sym_Case2} with the additional theoretical line according to the leading term of Theorem \ref{thm:exact_one}
              for $a=1.$}
          \end{figure}
      \end{minipage}
      \hspace{0.05\linewidth}
      \begin{minipage}{0.45\linewidth}
          \begin{figure}[H]
              \includegraphics[width=\linewidth]{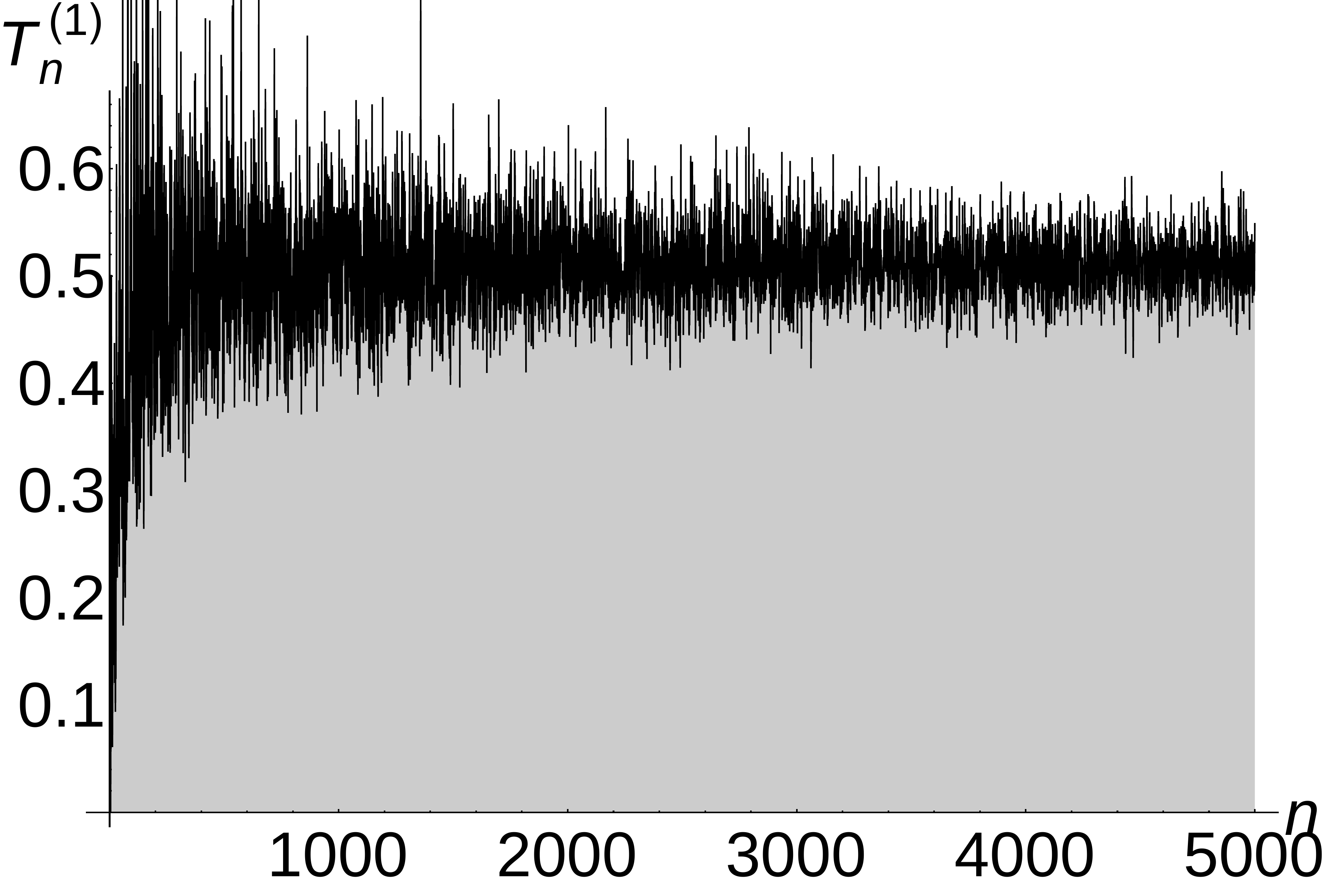}
              \caption{\label{fig:alg2} $T_n^{(1)}=\Theta(1)$ of Algorithm \ref{sym_Case1}.}
          \end{figure}
      \end{minipage}
  \end{minipage}

 \begin{minipage}{\linewidth}
      \centering
      \begin{minipage}{0.45\linewidth}
          \begin{figure}[H]
              \includegraphics[width=\linewidth]{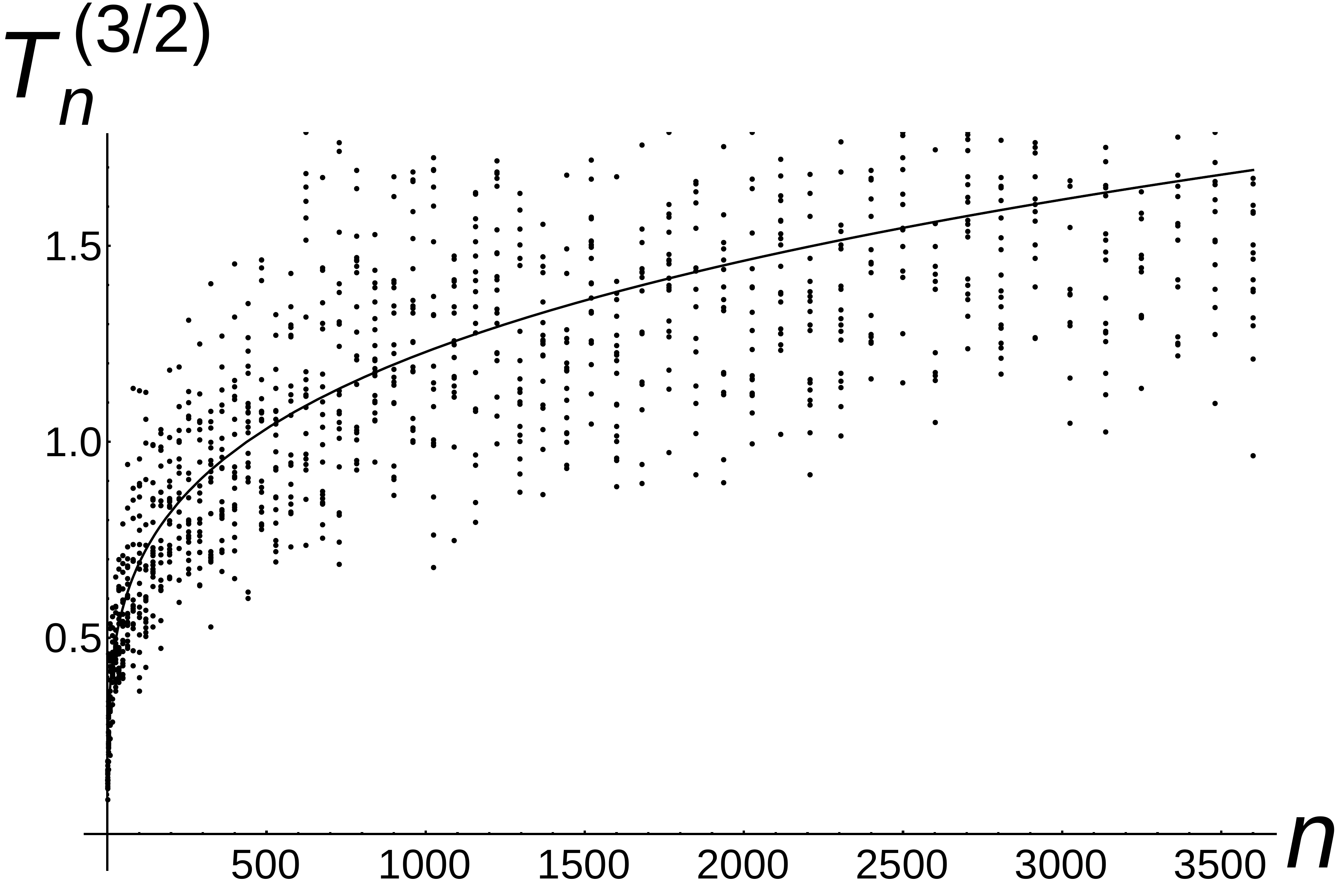}
              \caption{\label{fig:alg32} $T_n^{(3/2)}\sim \frac{\Gamma\left(\frac{7}{4}\right)}{2^{\frac{3}{4}}\left(\frac{5}{2}\right)}n^{\frac{1}{4}} $ of Algorithm \ref{sym_Case2} with the additional theoretical 
              line according to the leading term of Theorem \ref{thm:exact_one} for $a=3/2.$}
          \end{figure}
      \end{minipage}
      \hspace{0.05\linewidth}
      \begin{minipage}{0.45\linewidth}
          \begin{figure}[H]
              \includegraphics[width=\linewidth]{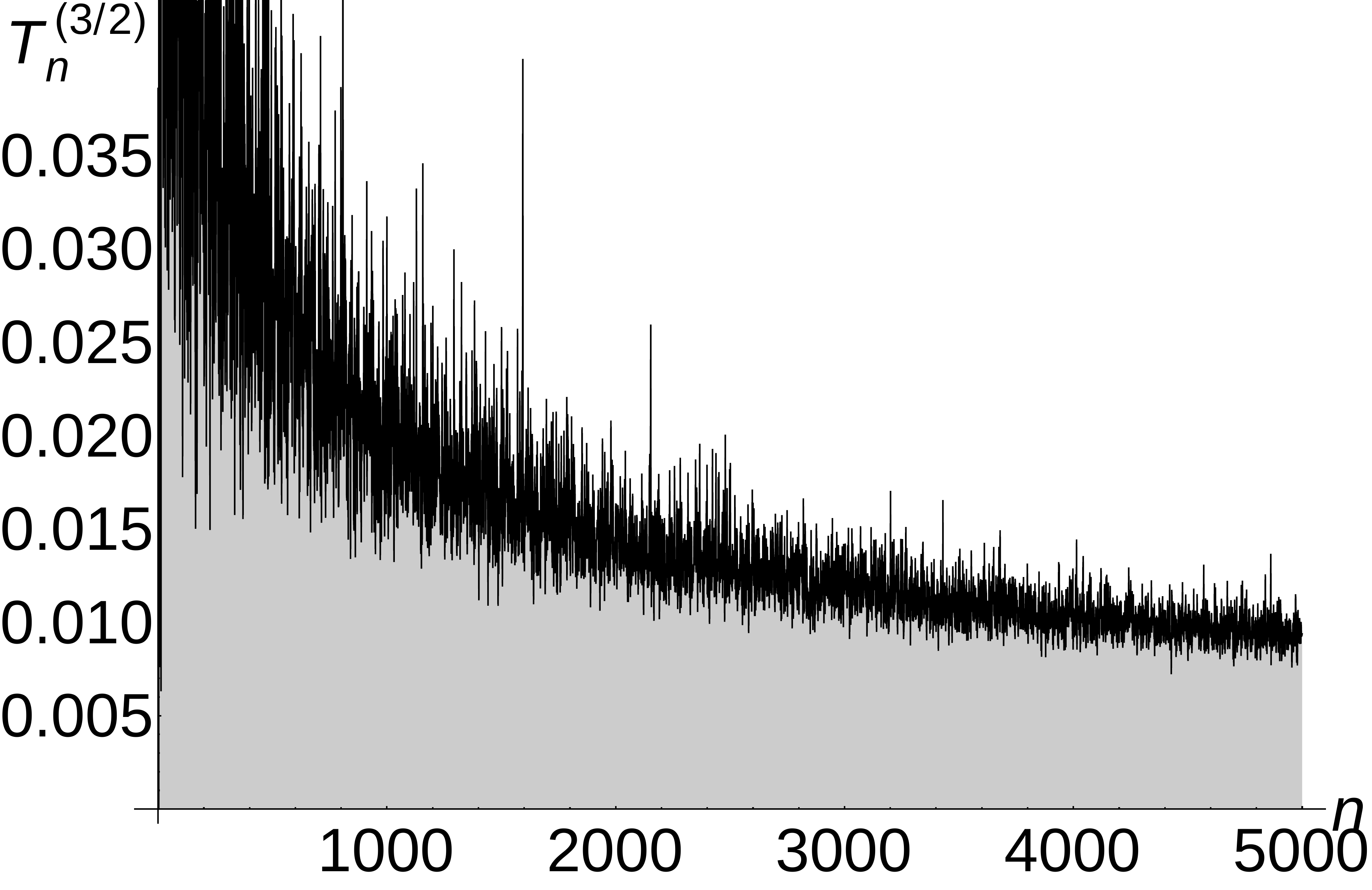}
              \caption{\label{fig:alg32b} $T_n^{(3/2)}=\Theta\left(\frac{1}{\sqrt{n}}\right)$ of Algorithm \ref{sym_Case1}.}
          \end{figure}
      \end{minipage}
  \end{minipage}

 \begin{minipage}{\linewidth}
      \centering
      \begin{minipage}{0.45\linewidth}
          \begin{figure}[H]
              \includegraphics[width=\linewidth]{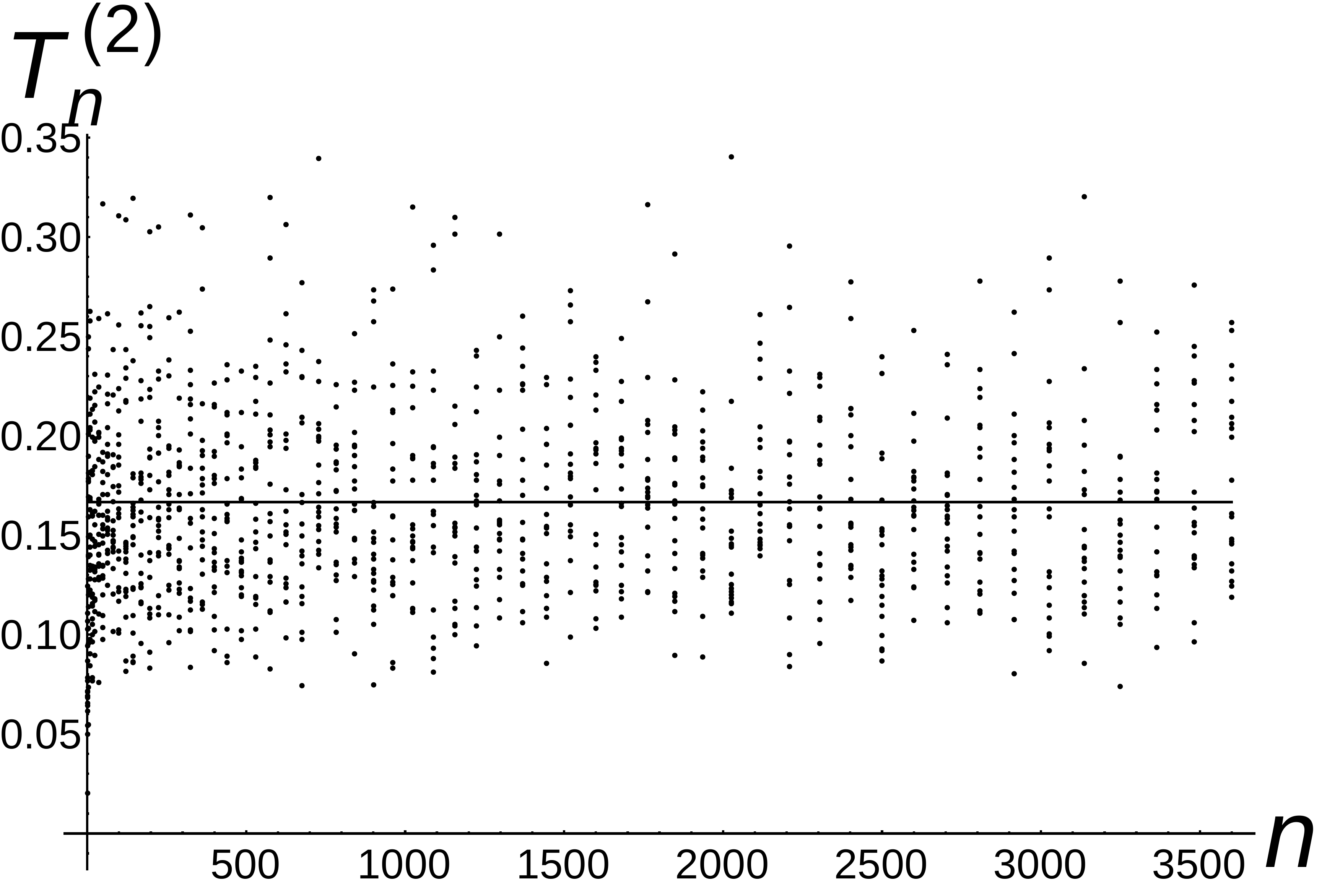}
              \caption{\label{fig:alg3} $T_n^{(2)}\sim \frac{1}{6}$ of Algorithm \ref{sym_Case2} with the additional theoretical line according to the leading term of Theorem \ref{thm:exact_one} for $a=2.$}
          \end{figure}
      \end{minipage}
      \hspace{0.05\linewidth}
      \begin{minipage}{0.45\linewidth}
          \begin{figure}[H]
              \includegraphics[width=\linewidth]{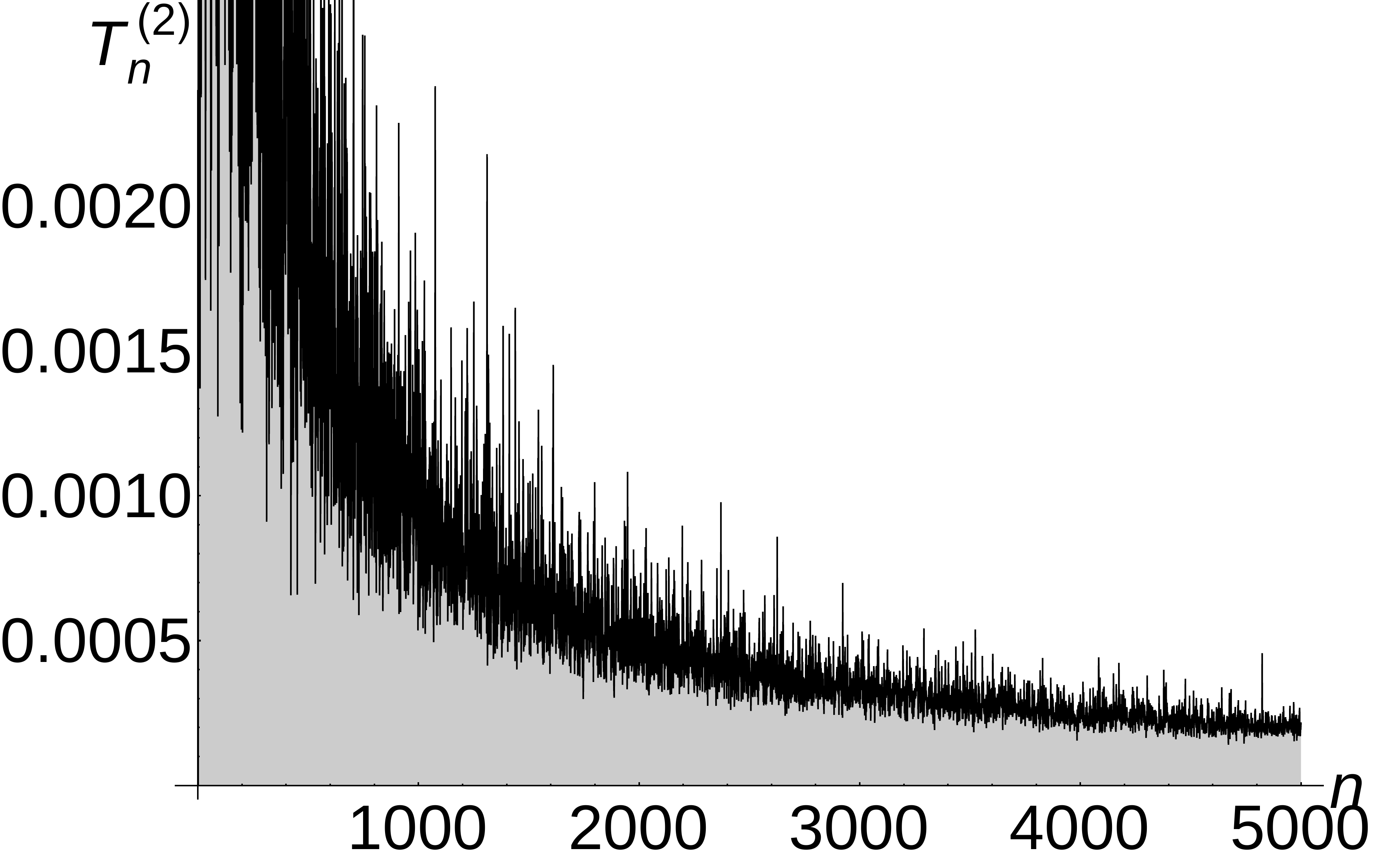}
              \caption{\label{fig:alg4} $T_n^{(2)}=\Theta\left(\frac{1}{n}\right)$ of Algorithm \ref{sym_Case1}.}
          \end{figure}
      \end{minipage}
  \end{minipage}

\section{Conclusion and Future Direction}
\label{sec:con}
 In this paper the following natural problem was investigated: given $n$ uniformly random mobile sensors in $m$-dimensional unit cube, where $m\in\{1,2\},$ what is the minimal energy consumption to move them
so that are pairwise at interference distance at least $s$ apart, and so that every point of $m$-dimensional unit cube
is within the range of at least one sensor?

As energy consumption measure for the displacement of $n$ sensors we considered the $a$-total displacement defined as the sum $\sum_{i=1}^n d_i^a,$ where $d_i$ is the distance sensor $i$ has been moved
and $a>0.$
The main findings can be summarized as follows:
\begin{itemize}
\item For the sensors placed on the unit interval, sensing radius $r_1=\frac{1}{2n}$ and interference distance $s=\frac{1}{n}$ the expected minimal $a$-total displacement is of order $O\left(n^{1-\frac{a}{2}}\right).$
When $r_1=\frac{1+\epsilon}{2n}$ and $s=\frac{1-\delta}{n},$ provided that $\epsilon>0$, $1>\delta>0$ are arbitrary small constants independent on the number of sensors $n,$
then there is an algorithm with $O\left(n^{1-a}\right)$ expected $a$-total displacement for all powers $a>0.$
\item For the case of the unit square and $a>0$, square sensing radius $r_2=\frac{1}{2\sqrt{n}}$ and interference distance $s=\frac{1}{\sqrt{n}}$ the expected minimal $a$-total displacement is at least of order 
$\Omega\left(\left(\log(n)\right)^{\frac{a}{2}}n^{1-\frac{a}{2}}\right),$ provided that $n$ is the square of a natural number.
When $r_2=\frac{1+\epsilon}{2\lfloor\sqrt{n}\rfloor}$ and $s=\frac{1-\delta}{\lfloor\sqrt{n}\rfloor},$ provided that $\epsilon>0$, $1>\delta>0$ are arbitrary small constants independent on the number of sensors $n,$
then there is an algorithm with $O\left(n^{1-\frac{a}{2}}\right)$ expected $a$-total displacement for all powers $a\ge 1.$
\end{itemize}

This paper opens several research directions. 

First, it would be interesting to know what happens if $\epsilon$ and $\delta$ depend on $n$  and decreases to $0.$ This would give the complete picture of the threshold phenomena
for \textit{coverage \& interference} requirement.

Second, in this paper we investigated \textit{coverage \& interference} requirement only for one and two dimensional network. It is an open problem to generalize this study to the higher dimensions
and investigate threshold phenomena for th $m$-dimensional cube, similar to $1$- and $2$-dimensional cubes.

Additionally it would be interesting for future research to study \textit{coverage \& interference} requirement for non-uniform displacement of sensors, on other domains,
as well for some real-life sensor displacement.

\bibliographystyle{abbrvnat}
\bibliography{refs}
\label{sec:biblio}
\newpage
\appendix
\section*{Appendix A} 
\textbf{Let us recall Lemma \ref{lemma_e}.}

Let $a>0$ be a constant.
Fix $\gamma>0$ independent on $n.$ Let $\rho=\frac{1+\gamma}{n}.$
Assume that $\ell, n$ are positive integers and $\ell\le n.$
Then
\begin{equation}
\mathbb{E}\left[\left(|\mathrm{Beta}(\ell,n-\ell+1)-\rho \ell|^{+}\right)^a\right]=O\left(\frac{1}{n^a}\right),
\,\,\,\text{uniformly in}\,\,\,\ell\in\{1,2,\dots,n\},
\end{equation}
\begin{equation}
\sum_{\ell=1}^{n}\frac{n}{\ell} \mathbb{E}\left[\left(|\mathrm{Beta}(\ell,n-\ell+1)-\rho \ell|^{+}\right)^a\right]=O\left(n^{1-a}\right).
\end{equation}

\begin{proof} (\textbf{Lemma \ref{lemma_e}})
Let $b=\lceil a\rceil$ be the smallest integer greater or equal to $a.$
We estimate separately when
$0\le \rho \ell\le1- \frac{2}{n+b-1}$ and when $1-\frac{2}{n+b-1}< \rho \ell\le 1.$

\textit{Case $0\le \rho \ell\le1- \frac{2}{n+b-1}.$ } Observe that
\begin{equation}
\label{eq:asak100a}
\mathbb{E}\left[\left(\left|\mathrm{Beta}(\ell,n-\ell+1)-\rho \ell\right|^{+}\right)^b\right]=\int_{\rho \ell}^{1}(t-\rho \ell)^bf_{\ell,n}(t)dt\le \int_{\rho \ell}^{1}t^bf_{\ell,n}(t)dt,
\end{equation}
where $f_{\ell,n}(t)=\ell\binom{n}{\ell}t^{\ell-1}(1-t)^{n-\ell}.$
Applying Identities (\ref{incomplete:first}), (\ref{probal_eq}) for $c:=\ell+b,\,\,$ $d:=n-\ell+1,$ $z:=1$ 
and $c:=\ell+b,\,\,$ $d:=n-\ell+1,$ $z:=\rho \ell$ 
we have
\begin{align}
 \nonumber\int_{\rho \ell}^{1}&t^bf_{\ell,n}(t)dt=\int_{0}^{1}t^bf_{\ell,n}(t)dt-\int_0^{\rho \ell}t^bf_{\ell,n}(t)dt\\ 
 \nonumber&=\frac{\ell(\ell+1)\dots(\ell+b-1)}{(n+1)(n+2)\dots(n+b)}\left(I_1(\ell+b,n-\ell+1)-I_{\rho \ell}(\ell+b,n-\ell+1)\right)\\
 \nonumber&=\frac{\ell(\ell+1)\dots(\ell+b-1)}{(n+1)(n+2)\dots(n+b)}\left(1-I_{\rho \ell}(\ell+b,n-\ell+1)\right)\\
 \nonumber &=\frac{\ell(\ell+1)\dots(\ell+b-1)}{(n+1)(n+2)\dots(n+b)}\sum_{j=0}^{\ell+b-1}\binom{n+b}{j}(\rho \ell)^j(1-\rho \ell)^{n+b-j}\\
 \nonumber &=\frac{\ell(\ell+1)\dots(\ell+b-1)}{(n+1)(n+2)\dots(n+b)}\times\\
 \label{eq:statistica01}&\sum_{j=0}^{\ell+b-1}\frac{n+b}{n+b-j}(1-\rho \ell)\binom{n+b-1}{j}(\rho \ell)^j(1-\rho \ell)^{n+b-1-j}.
\end{align}
From Inequality (\ref{eq:lecam}) for $x:=\rho \ell$ and $n:=n+b-1$ we get
\begin{align}
\nonumber&\binom{n+b-1}{j}(\rho \ell)^j(1-\rho \ell)^{n+b-1-j}\\
\label{eq:lecam_app}& \le \left(\frac{n+b-1}{(n+b-1)(1-\rho \ell)-1}\right)^{\frac{1}{2}} e^{-(n+b-1)\rho \ell}\frac{\left((n+b-1)\rho \ell\right)^j}{j!}.
\end{align}
Using assumption $\rho \ell\le 1-\frac{2}{n+b-1}$ we easily derive
\begin{equation}
\label{eq:lecam_estime} 
(1-\rho \ell) \left(\frac{n+b-1}{(n+b-1)(1-\rho \ell)-1}\right)^{\frac{1}{2}}\le\left(\frac{1-\rho \ell}{1-\rho \ell-\frac{1}{n+b-1}}\right)^{\frac{1}{2}}\le\sqrt{2}.
\end{equation}
Since $\rho \ell<1$ and $\rho =\frac{1+\gamma}{n},$ we have
\begin{align}
\nonumber&\frac{n+b}{n+b-j}\le\frac{n+b}{n+1-\ell}<\frac{n+b}{n+1-\frac{1}{\rho}}
\label{eq:lecam_ddd} 
=\frac{n+b}{n+1-\frac{n}{1+\gamma}}=\frac{n+b}{n\frac{\gamma}{1+\gamma}+1},\\
&\,\,\,\text{when}\,\,\,j\le \ell+b-1. 
\end{align}
Combining together (\ref{eq:asak100a}---\ref{eq:lecam_ddd}) 
we get
\begin{align}
\nonumber&\mathbb{E}\left(\left[|\mathrm{Beta}(\ell,n-\ell+1)-\rho \ell|^{+}\right)^b\right]\le \frac{\ell(\ell+1)\dots(\ell+b-1)}{(n+1)(n+2)\dots(n+b-1)}
\label{eq:fusolala}\times\\
&\times \frac{\sqrt{2}}{n\frac{\gamma}{1+\gamma}+1}e^{-(n+b-1)\rho \ell}\sum_{j=0}^{\ell+b-1}\frac{((n+b-1)\rho \ell)^j}{j!}.
\end{align}
Putting together assumptions: $j\le \ell+b-1$ and $\ell<n$ with the elementary inequality $\left(1+\frac{1}{x}\right)^x\le e,$ when $x>0$ we have
$$
\left(\frac{n+b-1}{n}\right)^j\le \left(\frac{n+b-1}{n}\right)^{n+b-1}=\left(\left(1+\frac{b-1}{n}\right)^{\frac{n}{b-1}}\right)^{\frac{(b-1)(n+b-1)}{n}}
\le e^{(b-1)b}.
$$
Hence
\begin{equation}
\label{eq:dusol}
(n+b-1)^j\le n^j e^{(b-1)b}.
\end{equation}
Observe  that
\begin{equation}
\label{eq:dusola}
e^{-(n+b-1)\rho \ell}\le e^{-n\rho \ell}.
\end{equation}
Combining together (\ref{eq:fusolala}---\ref{eq:dusola})
we get
\begin{align}
\nonumber&\mathbb{E}\left(\left[|\mathrm{Beta}(\ell,n-\ell+1)-\rho \ell|^{+}\right)^b\right]\le \frac{\ell(\ell+1)\dots(\ell+b-1)}{(n+1)(n+2)\dots(n+b-1)}\times\\
\label{eq:firstveldosa}&\times \frac{\sqrt{2}e^{(b-1)b}}{n\frac{\gamma}{1+\gamma}+1}e^{-n\rho \ell}\sum_{j=0}^{\ell+b-1}\frac{(n\rho \ell)^j}{j!}.
\end{align}
Using assumption $\rho n>1$ we easily derive the following inequality
\begin{equation}
\label{eq:desteer}
\frac{(n\rho \ell)^j}{j!}\le \frac{(n\rho \ell)^{j+1}}{(j+1)!},\,\,\,\text{when}\,\,\,j\le \ell-1.
\end{equation}
Hence
\begin{equation}
\label{eq;durenek}
\sum_{j=0}^{\ell}\frac{(n\rho \ell)^j}{j!}\le(\ell+1)\frac{(n\rho \ell)^\ell}{\ell!}.
\end{equation}
Observe that
\begin{equation}
\label{eq:asamutra}
\sum_{j=\ell+1}^{\ell+b-1}\frac{(n\rho \ell)^j}{j!}\le(b-1)\frac{(n\rho \ell)^{\ell+b-1}}{\ell!}.
\end{equation}
From Stirling's formula (\ref{eq:stirlingform}) for $N:=\ell$ we have
\begin{equation}
\label{eq:asak101a}
\frac{\ell^\ell}{\ell!}\le\frac{e^\ell}{\ell^{\frac{1}{2}}}\le e^\ell.
\end{equation}
Putting together (\ref{eq:firstveldosa})---(\ref{eq:asak101a}) 
we have
\begin{align*}
\mathbb{E}&\left[\left(|\mathrm{Beta}(\ell,n-\ell+1)-\rho \ell|^{+}\right)^b\right]\le\frac{\sqrt{2}e^{(b-1)b}\ell(\ell+1)\dots(\ell+b-1)}{(n+1)(n+2)\dots(n+b-1)\left(n\frac{\gamma}{1+\gamma}+1\right)}\times\\
&\times\left((\ell+1)+(b-1)\ell^{b-1}(n\rho)^{b-1}\right)\left(\frac{n\rho e}{e^{n\rho}}\right)^\ell.
\end{align*}
Since $\rho n=1+\gamma$ is some constant independent on $n$ we derive
\begin{equation}
\label{eq:firstveld}
\mathbb{E}\left[\left(|\mathrm{Beta}(\ell,n-\ell+1)-\rho \ell|^{+}\right)^b\right]\le\frac{O\left(\ell^{\max(b+1,2b-1)}\right)}{O\left(n^b\right)} \left(\frac{n\rho e}{e^{n\rho}}\right)^\ell.
\end{equation}
Let us recall that $b=\lceil a\rceil$ is the smallest integer greater or equal to $a.$
From Jensen's inequality for $f(x):=x^{\frac{\lceil a\rceil}{a}}$ and
$X:=\left(|\mathrm{Beta}(\ell,n-\ell+1)-\rho \ell|^{+}\right)^a$ we get
\begin{equation}
\label{eq:jensen_modyf}
\mathbb{E}\left[\left(|\mathrm{Beta}(\ell,n-\ell+1)-\rho \ell|^{+}\right)^a\right]\le \left(\mathbb{E}\left[\left(|\mathrm{Beta}(\ell,n-\ell+1)-\rho \ell|^{+}\right)^{\lceil a\rceil}\right]\right)^{\frac{a}{\lceil a\rceil}}.
\end{equation}
Putting together Estimation (\ref{eq:firstveld}), as well as $b=\lceil a\rceil$ and Inequality (\ref{eq:jensen_modyf}) we have
\begin{equation}
\label{eq:firstveldosos}
\mathbb{E}\left[\left(|\mathrm{Beta}(\ell,n-\ell+1)-\rho \ell|^{+}\right)^a\right]\le\frac{O\left(\ell^{\max\left(a+\frac{a}{\lceil a\rceil},2a-\frac{a}{\lceil a\rceil}\right)}\right)}{O\left(n^a\right)} 
\left(\left(\frac{n\rho e}{e^{n\rho}}\right)^{\frac{a}{\lceil a\rceil}}\right)^\ell.
\end{equation}
Combining assumption $\rho n=1+\gamma>1$ with the elementary inequality $\gamma+1<e^{\gamma},$ when $\gamma>0$ we deduce that
$\frac{n\rho e}{e^{n\rho}}=\frac{\gamma+1 }{e^{\gamma}}<1$. Hence
$$
\left(\frac{n\rho e}{e^{n\rho}}\right)^{\frac{a}{\lceil a\rceil}}\le 1.
$$
Therefore
\begin{equation}
\label{eq:dudek_marszowy01a}
\frac{O\left(\ell^{\max\left(a+\frac{a}{\lceil a\rceil},2a-\frac{a}{\lceil a\rceil}\right)}\right)}{O\left(n^a\right)} 
\left(\left(\frac{n\rho e}{e^{n\rho}}\right)^{\frac{a}{\lceil a\rceil}}\right)^\ell=O\left(\frac{1}{n^a}\right)\,\,\,\text{uniformly in}\,\,\,\ell\in\{1,2,\dots,n\}
\end{equation}
\begin{equation}
\label{eq:kutasos}
\sum_{\ell=1}^{n}\frac{n}{\ell}\frac{O\left(\ell^{\max\left(a+\frac{a}{\lceil a\rceil},2a-\frac{a}{\lceil a\rceil}\right)}\right)}{O\left(n^a\right)} 
\left(\left(\frac{n\rho e}{e^{n\rho}}\right)^{\frac{a}{\lceil a\rceil}}\right)^\ell =O\left(n^{1-a}\right).
\end{equation}
Putting together (\ref{eq:firstveldosos}), (\ref{eq:dudek_marszowy01a}) and (\ref{eq:kutasos}) we have
\begin{equation}
\label{eq:dudek_marszowy01}
\mathbb{E}\left[\left(|\mathrm{Beta}(\ell,n-\ell+1)-\rho \ell|^{+}\right)^a\right]=O\left(\frac{1}{n^a}\right),
\,\,\,\text{uniformly in}\,\,\,\ell\in\{1,2,\dots,n\},
\end{equation}
\begin{equation}
\label{eq:dudek_marszowy02}
\sum_{\ell=1}^{n}\frac{n}{\ell} \mathbb{E}\left[\left(|\mathrm{Beta}(\ell,n-\ell+1)-\rho \ell|^{+}\right)^a\right]=O\left(n^{1-a}\right).
\end{equation}
Finally, together (\ref{eq:dudek_marszowy01}) and (\ref{eq:dudek_marszowy02}) are enough to establish the first case.

\textit{Case  $1-\frac{2}{n+b-1}< \rho \ell\le 1.$} 
Observe that
\begin{align}
\nonumber\mathbb{E}&\left[\left(\left|\mathrm{Beta}(\ell,n-\ell+1)-\rho \ell\right|^{+}\right)^a\right]=\int_{\rho \ell}^{1}(t-\rho \ell)^af_{\ell,n}(t)dt
 \le\int_{\rho \ell}^{1}(1-\rho \ell)^af_{\ell,n}(t)dt\\
 &\label{eq:asak100b}\le\left(\frac{2}{n+b-1}\right)^a\int_{\rho \ell}^{1}f_{\ell,n}(t)dt.
\end{align}
Since $f_{\ell,n}(t)$ is the probability density function of the $\mathrm{Beta}(\ell,n-\ell+1),$ we have
\begin{equation}
\label{eq:probbbn}
\int_{\rho \ell}^{1}f_{\ell,n}(t)dt\le \int_{0}^{1}f_{\ell,n}(t)dt=1.
\end{equation}
Putting together (\ref{eq:asak100b}) and (\ref{eq:probbbn}) we have
\begin{equation}
\label{eq:des}
\mathbb{E}\left[\left(|\mathrm{Beta}(\ell,n-\ell+1)-\rho \ell|^{+}\right)^a\right]=O\left(\frac{1}{n^a}\right),
\,\,\,\text{uniformly in}\,\,\,\ell\in\{1,2,\dots,n\}.
\end{equation}
Since $1-t\le1-\rho \ell<\frac{2}{n+b-1}$ and $t\le 1$, we have $t^{\ell-1}\le 1$ and $(1-t)^{n-\ell}<\left(\frac{2}{n+b-1}\right)^{n-\ell}.$ Putting all this together 
with the elementary inequality $\left(1+\frac{1}{x}\right)^x\le e,$ when $x>0$ we have
\begin{align}
\label{eq:probbban}
\nonumber\sum_{\ell=1}^n\frac{1}{\ell}\int_{\rho \ell}^{1}f_{\ell,n}(t)dt&\le \sum_{\ell=1}^{n}\binom{n}{\ell}\left(\frac{2}{n+b-1}\right)^{n-\ell}\int_{\rho \ell}^{1}dt\le\left(1+\frac{2}{n+b-1}\right)^n\\
&\le\left(\left(1+\frac{2}{n+b-1}\right)^{\frac{n+b-1}{2}}\right)^{\frac{2n}{n+b-1}}
\le e^{\frac{2n}{n+b-1}}=O(1).
\end{align}
Together (\ref{eq:asak100b}) and (\ref{eq:probbban}) imply
\begin{align}
\nonumber\sum_{\ell=1}^{n}&\frac{n}{\ell} \mathbb{E}\left[\left(|\mathrm{Beta}(\ell,n-\ell+1)-\rho \ell|^{+}\right)^a\right]\\
&\le n\left(\frac{2}{n+b-1}\right)^a \sum_{\ell=1}^n\frac{1}{\ell}\int_{\rho \ell}^\ell f_{\ell,n}(t)dt
\label{eq:oszac}=O\left(n^{1-a}\right).
\end{align}
Finally, (\ref{eq:des}) and (\ref{eq:oszac})
are enough to prove the second case and sufficient to complete the proof of Lemma \ref{lemma_e}. 
\end{proof}
\section*{Appendix B}
\textbf{Let us recall Lemma \ref{lemma_f}.}

Let $a>0$ be a constant.
Fix $1>\delta>0$ independent on $n.$ Let $s=\frac{1-\delta}{n}.$
Assume that $\ell, n$ are positive integers and $\ell\le n.$
Then
\begin{equation}
\sum_{\ell=1}^{n}\frac{n}{\ell} \mathbb{E}\left[\left(|s\ell-\mathrm{Beta}(\ell,n-\ell+1)|^{+}\right)^a\right]=O\left(n^{1-a}\right).
\end{equation}

\begin{proof} (\textbf{Lemma \ref{lemma_f}})
First of all observe that
\begin{equation}
\label{eq:asak100aw}
\mathbb{E}\left[\left(\left|s\ell-\mathrm{Beta}(\ell,n-\ell+1)\right|^{+}\right)^a\right]=\int_{0}^{s\ell}(s\ell-t)^af_{\ell,n}(t)dt\le (s\ell)^a\int_{0}^{s\ell}f_{\ell,n}(t)dt,
\end{equation}
where $f_{\ell,n}(t)=\ell\binom{n}{\ell}t^{\ell-1}(1-t)^{n-\ell}.$
Applying Identities (\ref{incomplete:first}), (\ref{probal_eq}), (\ref{eq:binole}) for $c:=\ell,\,\,$ $d:=n-\ell+1$ and $z:=s\ell$ we have
\begin{equation}
\label{eq:statistica01a}
\int_{0}^{s\ell}f_{\ell,n}(t)dt=\sum_{j=\ell}^{n}\binom{n}{j}(s \ell)^j(1-s \ell)^{n-j}.
\end{equation}
From Inequality (\ref{eq:lecam}) for $x:=s\ell$ we get
\begin{equation}
\label{eq:lecam_app1m}
\binom{n}{j}(s\ell)^j(1-s\ell)^{n-j}\le
\left(\frac{n}{n(1-s \ell)-1}\right)^{\frac{1}{2}} e^{-ns\ell}\frac{(ns\ell)^j}{j!}.
\end{equation}
Using assumption $s \ell<1-\delta$ we easily derive
\begin{equation}
\label{eq:lecam_estime1b} 
\left(\frac{n}{n(1-s \ell)-1}\right)^{\frac{1}{2}}\le\left(\frac{1}{\delta-\frac{1}{n}}\right)^{\frac{1}{2}}\le\sqrt{\frac{2}{\delta}},\,\,\,\text{when}\,\,\,n>2/\delta.
\end{equation}
Combining together (\ref{eq:asak100aw}---\ref{eq:lecam_estime1b}) 
we get
\begin{align}
\nonumber\mathbb{E}&\left[\left(|s\ell-\mathrm{Beta}(\ell,n-\ell+1)|^{+}\right)^a\right]\le (s\ell)^a \sqrt{\frac{2}{\delta}}e^{-ns\ell}\sum_{j=\ell}^{n}\frac{(ns \ell)^j}{j!}\\
\label{eq:firstveldos}&\le (s\ell)^a \sqrt{\frac{2}{\delta}}e^{-ns\ell}\sum_{j=\ell}^{\infty}\frac{(ns \ell)^j}{j!},
\,\,\,\text{when}\,\,\,n>2/\delta.
\end{align}
Using assumption $sn<1$ we can easily derive the following inequality
$$
\frac{(ns\ell)^{j}}{j!}\ge \frac{(ns\ell)^{j+1}}{(j+1)!},\,\,\,\,\,\,\text{when}\,\,\,\,\,\,j\ge \ell-1.
$$
Therefore
\begin{equation}
\nonumber
\sum_{j=\ell}^{\infty}\frac{(ns\ell)^j}{j!}=\sum_{j=\ell}^{\lceil le\rceil }\frac{(ns\ell)^j}{j!}+\sum_{j=\lceil le\rceil+1 }^{\infty}\frac{(ns\ell)^j}{j!}
\label{eq:styrly02} \le \frac{(ns\ell)^\ell}{\ell!}(\ell e+1)+\sum_{j=\lceil le\rceil+1 }^{\infty}\frac{(ns\ell)^j}{j!}.
\end{equation}
Applying Stirling's formula (\ref{eq:stirlingform}) for $N:=\ell$ and $N:=j$ we get
$$
\frac{\ell^\ell}{\ell!}\le\frac{e^\ell}{\ell^{\frac{1}{2}}}\le e^\ell,\,\,\,\,\,\,   \frac{1}{j!}\le\frac{e^j}{j^{j+\frac{1}{2}}}\le \frac{e^j}{j^j}.
$$
Using these estimations in Inequality (\ref{eq:styrly02})  we derive
\begin{equation}
\label{eq:ffinn}
\sum_{j=\ell}^{\infty}\frac{(ns\ell)^j}{j!}\le(nse)^\ell(\ell e+1)+\sum_{j=\lceil le\rceil+1 }^{\infty}\left(\frac{ns\ell e}{j}\right)^j.
\end{equation}
From assumption $sn<1$ we get
\begin{equation}
\label{eq:ffinna}
\sum_{j=\lceil le\rceil+1 }^{\infty}\left(\frac{ns\ell e}{j}\right)^j\le \sum_{j=\lceil le\rceil+1 }^{\infty}(ns)^j\le\sum_{j=\ell }^{\infty}(ns)^j=(ns)^\ell\frac{1}{1-ns}=\frac{(ns)^\ell}{\delta}.
\end{equation}
Together  Inequalities (\ref{eq:firstveldos}), (\ref{eq:ffinn}) and (\ref{eq:ffinna}) imply
\begin{align}
\nonumber\mathbb{E}&\left[\left(|s\ell-\mathrm{Beta}(\ell,n-\ell+1)|^{+}\right)^a\right]\\
&\le\sqrt{\frac{2}{\delta}} s^a\ell\sum_{\ell=1}^{n}\left(\left(\frac{nse}{e^{ns}}\right)^\ell(\ell e+1)\ell^{a-1}+\left(\frac{ns}{e^{ns}}\right)^\ell \frac{\ell^{a-1}}{\delta}\right),
\,\,\,\,\,\label{eq:ffinnb}\,\,\,\text{when}\,\,\,n>2/\delta.
\end{align}
Combining assumption $s n=1-\delta$ with the elementary inequalities: $1-\delta<e^{-\delta}$ and $1-\delta<e^{1-\delta},$ when $\delta\in (0,1)$ we deduce that
$\frac{nse}{e^{ns}}=\frac{1-\delta}{e^{-\delta}}<1$ and $\frac{ns}{e^{ns}}=\frac{1-\delta}{e^{1-\delta}}<1.$ Hence
\begin{equation}
\label{eq:ostatni02a}
\sum_{\ell=1}^{n}\left(\left(\frac{nse}{e^{ns}}\right)^\ell(\ell e+1)\ell^{a-1}+\left(\frac{ns}{e^{ns}}\right)^\ell\frac{\ell^{a-1}}{\delta} \right)=O(1).
\end{equation}
Putting together (\ref{eq:ffinnb}), (\ref{eq:ostatni02a}) and assumption $s n=1-\delta$ we conclude that
$$
\sum_{\ell=1}^{n}\frac{n}{\ell} \mathbb{E}\left[\left(|s\ell-\mathrm{Beta}(\ell,n-\ell+1)|^{+}\right)^a\right]=O\left(n^{1-a}\right).
$$
This concludes the proof of Lemma \ref{lemma_f}. 
\end{proof}
\end{document}